\documentclass[amssymb,amsfonts]{amsart}
\usepackage[T1]{fontenc}
\usepackage{svn}
\usepackage{amssymb, latexsym}
\usepackage{paralist}
\usepackage{srcltx}
\usepackage{tensor}
\usepackage{constants}
\usepackage{hyperref}

\SVN $Date: 2013-04-05 16:52:57 +0200 (Fri, 05 Apr 2013) $

\theoremstyle{plain}
\newtheorem{thm}[equation]{Theorem}
\newtheorem{prop}[equation]{Proposition}
\newtheorem{lem}[equation]{Lemma}
\newtheorem{cor}[equation]{Corollary}
\theoremstyle{definition}
\newtheorem{defn}[equation]{Definition}

\theoremstyle{remark}
\newtheorem{rmk}[equation]{Remark}

\newcommand{\msweyl}{\ensuremath{\mathcal{Q}}}
\newcommand{\msmax}{\ensuremath{\mathcal{B}}}
\newcommand{\spacetime}{\ensuremath{\mathcal{M}}}
\newcommand{\cauchy}{\ensuremath{\Sigma}}
\newcommand{\lcsym}{\ensuremath{\varepsilon}} 
\newcommand{\horizon}{\ensuremath{\mathfrak{h}}}
\newcommand{\outercom}{\ensuremath{\mathcal{E}}}
\newcommand{\maxw}{\ensuremath{\mathcal{H}}}
\newcommand{\preernst}{\ensuremath{\hat{\mathcal{F}}}}
\newcommand{\ernstform}{\ensuremath{\mathcal{F}}}
\newcommand{\maxpot}{\ensuremath{\Xi}}
\newcommand{\maxpotr}{\ensuremath{P_0}}
\newcommand{\ernstpot}{\ensuremath{\sigma}}
\newcommand{\msmaxpot}{\ensuremath{V}}
\newcommand{\Ricci}{\ensuremath{\mathit{Ric}}}
\newcommand{\Riemann}{\ensuremath{\mathit{Riem}}}
\newcommand{\weyls}{\ensuremath{\mathcal{C}}}
\newcommand{\Iasd}{\ensuremath{\mathcal{I}}}
\newconstantfamily{errorterms}{symbol=\mathfrak{e}}

\newcommand{\errsl}[1]{\ensuremath{\Cl[errorterms]{#1}}}
\newcommand{\errsr}[1]{\ensuremath{\Cr{#1}}}
\newcommand{\lbar}{\ensuremath{\underline{l}}}
\newcommand{\ernstnd}{\ensuremath{f}}
\newcommand{\covD}{\ensuremath{\nabla}}
\newcommand{\lieD}{\ensuremath{\pounds}}

\newcommand{\Real}{\ensuremath{\mathbb{R}}}

\makeatletter

\newcommand{\Rmnum}[1]{\expandafter\@slowromancap\romannumeral #1@}
\makeatother

\numberwithin{equation}{subsection}

\begin{document}
\title[Multiple black holes]{Non-existence of multiple-black-hole solutions
close to Kerr-Newman}
\thanks{Version as of \SVNDate}
\author[W. W.-Y. Wong]{Willie Wai-Yeung Wong}
\address{\'Ecole Polytechnique F\'ed\'erale de Lausane, Lausanne,
Switzerland}
\thanks{W.W.Wong's research for this paper was supported in part by 
the Commission of the European Communities, ERC Grant Agreement No
208007, at the Department of Pure Mathematics and Mathematical
Statistics, University of Cambridge, Cambridge, UK; and in part by the
Swiss National Science Foundation.}
\email{willie.wong@epfl.ch}
\author[P. Yu]{Pin Yu}
\address{Mathematical Sciences Center, Tsinghua University, Beijing, China}
\thanks{P.Yu's research for this paper was supported in part by
Princeton University, Princeton, NJ, USA, and in part by NSF-China
Grant 11101235.}
\email{pin@math.tsinghua.edu.cn}
\thanks{The authors would like to thank Piotr Chru\'sciel for his
cogent comments on the manuscript, as well as the two anonymous
referees for their suggestions.}

\begin{abstract}
We show that a stationary asymptotically flat electro-vacuum solution
of Einstein's equations that is everywhere locally ``almost
isometric'' to a Kerr-Newman solution cannot admit more than one event
horizon. Axial symmetry is not assumed. In particular this implies
that the assumption of a single event horizon in
Alexakis-Ionescu-Klainerman's proof of perturbative uniqueness of
Kerr black holes is in fact unnecessary. 
\end{abstract}

\maketitle
\tableofcontents

\section{Introduction}
The goal of the present paper is to provide a justification for the
intuitively obvious fact that
\begin{quote}
\emph{A stationary electro-vacuum space-time that is everywhere \emph{almost}
isometric to Kerr-Newman can admit at most a single event horizon.}
\end{quote}
Roughly speaking, we do not expect small perturbations of the metric
structure to allow the topology (of the domain of outer
communications) of the solution to change greatly. Or, slightly
differently put, we expect that Weyl's observation for 
multiple-static-black-hole solutions remain true in the stationary case, that 
along the axes connecting the multiple black holes, the local geometry should 
be very different from what is present in a Kerr-Newman solution. In 
practice, however, one needs to be specific about what \emph{almost 
isometric} means. This shall be described later in this introduction. As a
direct consequence of the main result from this paper, we can slightly
improve the main theorem of Alexakis-Ionescu-Klainerman
\cite{AlIoKl2010} to remove from it the assumption that the space-time
only has one bifurcate event horizon. A secondary consequence of the
current paper is that it casts some new light on the tensorial
characterisations of Kerr and Kerr-Newman space-times due to Mars
\cite{Mars1999} and the first author \cite{Wong2009}. 

\subsection{History and overview}
The greater setting in which this paper appears is the study of the
``black hole uniqueness theorem''. Prosaically stated, the theorem
claims that
\begin{quote}
\emph{The only nondegenerate stationary\footnote{Admitting a Killing vector field
that becomes the time-translation at spatial infinity} electro-vacuum asymptotically flat
space-times are described by the three-parameter Kerr-Newman family.}
\end{quote}
The nondegeneracy here refers to conditions on the geometry of the
event horizon, or constraints on some asymptotic constants, or both,
of the solution. That a certain nondegeneracy is required is already
necessitated by the existence of the Majumdar--Papapetrou solutions
(see, e.g.~\cite{HarHaw1972}), which represent static multiple-black-hole
solutions in which the gravitational attraction between the black holes are
balanced out by their mutual electromagnetic repulsion. In the present
paper all black holes are nondegenerate or subextremal; as shall be
seen the argument depends strongly on the presence of nondegenerate
bifurcate event horizons. For the degenerate case we refer the readers
to \cite{HarHaw1972, ChrTod2007, ChrNgu2010, FigLuc2010, NeuHen2012} and the
references therein. 

The expectation that one such theorem may be available goes back at
least to Carter's lecture \cite{Carter1973}, where a first version of
a ``no hair'' theorem was proven; the hypotheses for this theorem
assumes, in particular, that the space-time is axisymmetric in
addition to being stationary. For \emph{static}\footnote{Admitting
a \emph{hypersurface-orthogonal} Killing vector field that is the time-translation at spatial
infinity} solutions a general uniqueness theorem was already
established without additional symmetry assumptions by Israel
\cite{Israel1967, Israel1968}. By appealing to Hawking's strong
rigidity theorem (see next paragraph), however, one can assume (with
some loss of generality) that any reasonable stationary
black-hole space-time is in fact axisymmetric. This additional
symmetry can be used to great effect: for the Kerr-Newman solutions
the stationary Killing field is not everywhere time-like due to the
presence of the ergoregions. Thus a symmetric reduction of Einstein's
equations with just a stationarity assumption (as opposed to a
staticity one) is insufficient to reduce the hyperbolic system of
equations to an elliptic one, for which uniqueness theorems are more
readily available (or widely known). With the additional axial
symmetry, the equations of motions for general relativity can be shown
to reduce to that of a harmonic map \cite{Buntin1983, Mazur1982,
Carter1985, Robins1975}, for which elliptic techniques (maximum
principle etc.) can be used to obtain the uniqueness result. For a
modern discussion one can consult Heusler's monograph
\cite{Heusle1996} in which various natural generalisations of this
method are considered. For some more historical notes and critical
analyses of these more classical results, see
\cite{Chrusc1994, Chrusc1996}. More recently, Costa in his PhD dissertation
\cite{Costa2010} gave a complete and modern derivation of the black
hole uniqueness theorem, in the formulation which is amenable to the
approach described above (namely first establishing axial symmetry and
then obtaining uniqueness using elliptic methods). 

One of the main shortfalls of the above approach is that Hawking's
rigidity theorem, as originally envisioned, requires that the
space-time be real analytic. Thus the result established for black
hole uniqueness is conditional on either the space-time being \emph{a
priori} axisymmetric, or real analytic. To overcome this problem,
Ionescu and Klainerman initiated a program to study the black hole
uniqueness problem as a problem of ``unique continuation''; namely,
one considers the \emph{ill-posed} initial value problem for the
Einstein equations with data given on the event horizon and try to
demonstrate a uniqueness property for the solution in the domain of
outer communications (outside the black hole; the problem of extending
to the \emph{inside} of the black hole, which does not suffer from the
obstruction of the ill-posedness of the initial value problem, has been 
considered before by other authors \cite{FrRaWa1999, Racz2001}). Their 
first approach to this problem \cite{IonKla2009, IonKla2009a} (see also the
generalisation by the first author \cite{Wong2009a}) provided a different
conditional black hole uniqueness result: instead of demanding the
space-time be axisymmetric or real analytic, the extra condition is
provided by, roughly speaking, prescribing the geometry of the event
horizon as an embedded null hypersurface in the space-time. Through
unique continuation, this boundary condition suffices to imply that
the so-called Mars-Simon tensor \cite{Mars1999, Wong2009} vanishes
everywhere, which shows that the exterior domain of the space-time is 
everywhere locally isometric to a Kerr(-Newman) black hole. A second
approach to this problem was later taken together with Alexakis
\cite{AlIoKl2010, AlIoKl2010a}, where under the assumption that the
Mars-Simon tensor is ``small'' one can extend Hawking's rigidity
theorem to the non-analytic case (see also the generalisation by the
second author \cite{Yu2010}). By appealing to the \emph{axisymmetric}
version of the black hole uniqueness theorem, this last theorem returns
us to a statement similar to Carter's original ``no hair'' theorem:
there are no other stationary electro-vacuum asymptotically flat
space-times in a small neighbourhood of the Kerr-Newman family. One of
the technical assumptions made in \cite{AlIoKl2010} is that the
space-time admits only one connected component of the event horizon;
in this paper we remove that assumption. 

The arguments described in the previous paragraph relied upon a
tensorial \emph{local} characterisation of the Kerr-Newman space-times
due to Mars and then to the first author \cite{Mars1999, Wong2009}. In
those two papers, that a region in a stationary solution to Einstein's
equations is locally isometric to a Kerr(-Newman) space-time is shown
to be equivalent to the vanishing of certain algebraic expressions
relating the Weyl curvature, the Ernst potential, the Ernst two form,
and the electromagnetic field. It is clear from the algebraic nature
of the expressions that if
the metric of a stationary solution and the electromagnetic field are
$C^{2}$ close to that of a Kerr-Newman space-time in local
coordinates, the algebraic expressions will also be suitably small.
The converse, however, is not obviously true: the demonstrations in
\cite{Mars1999, Wong2009} construct local coordinate systems by
first finding a holonomic frame field. Hence exact cancellations,
and not just approximate ones, are necessary to guarantee
integrability. As already was used in \cite{AlIoKl2010}, and
generalised further in the current paper, we show what can be
interpreted as a
partial converse. In particular, we show that one can reconstruct a
scalar function to serve as an analogue of the $r$ coordinate of Boyer-Lindquist presentation of the
Kerr-Newman metric, and thereby make use of many of its nice
properties. Critically used in \cite{AlIoKl2010} and
\cite{Yu2010} is that the level surfaces of this ``analogue-$r$'' have good
pseoduconvexity properties for a unique continuation argument; in this
paper we use the property that the ``analogue-$r$'' function behaves like the distance function
from a large sphere near infinity, and cannot have a critical point
outside the event horizons.

That analogues of the $r$ coordinate play important roles in
black hole uniqueness theorems is not new. They typically appear as
the inverse of the Ernst potential, and are used implicitly in
Israel's proofs for the static uniqueness theorems \cite{Israel1967,
Israel1968} (see also \cite{Robins1977, Simon1985, Alam1992}
which share some motivation with the present paper). Incidentally,
the proof by M\"uller zum Hagen and Seifert \cite{MulSei1973}
of non-existence of multiple black holes in the \emph{static}
axi-symmetric case also employs the properties of some analogue of
this $r$ function; whereas we (as will be indicated) use a
mountain-pass lemma to drive our non-existence proof, M\"uller zum
Hagen and Seifert employed a force balance argument that is somewhat
reminiscent of the recent work of Beig, Gibbons, and Schoen
\cite{BeGiSc2009}.

In the present paper we show that multiple stationary black hole configurations
cannot be possible were the solution everywhere (in the domain
of outer communications) locally close to, but not necessarily
isometric to, a subextremal Kerr-Newman solution. We would be 
remiss not to mention
the literature concerning the case where the ``smallness parameter''
of being close to Kerr-Newman solutions is replaced by the restriction
of axisymmetry (which, in particular, would apply assuming a smooth
version of Hawking's rigidity theorem is available. Note also that the
\emph{static} case behaves somewhat better; see previous paragraphs). 
On the one hand we have the construction
(see \cite{Weinst1990, Weinst1992, Weinst1996} and references therein) 
of solutions with multiple spinning black holes sharing the same axis of 
rotation, which may be singular along the axis (see also \cite{Nguyen2011} 
for an analysis of their regularity property). This construction uses
again the stationary and axial symmetries to reduce the question to
the existence of certain harmonic maps with boundary conditions
prescribed along the axis of symmetry and the event horizon. On the
other hand we also have the approach by studying the Ernst formulation of
Einstein's equations in stationary-axisymmetric case, and using the
inverse scattering method to obtain a non-existence result in the
two-body case; see \cite{NeuHen2009, NeuHen2012} and references therein. 
As the methods employed in the approaches mentioned above are rather
orthogonal to ours (for showing non-existence the general approach in
the stationary axisymmetric case is to show the \emph{lack of}
regularity along the axis connecting the multiple black holes), it is 
hard to compare the results obtained, especially in view of the fact
that the objects involved are not \emph{supposed} to exist as smooth
solutions. 

One last remark about the theorem proved in this paper. \emph{A
posteriori}, by combining the results of the present paper with
\cite{AlIoKl2010} and the axisymmetric uniqueness result of
\cite{Costa2010}, we have that the only space-times that satisfy our
hypotheses are in fact the Kerr-Newman solutions. Hence while it is
\emph{a priori} 
necessary to state our theorem and perform our computations in a way
that admits the possibility such additional non-Kerr-Newman solutions
exist, one should not try too hard to precisely imagine such
additional solutions. 

\subsection{Main idea of proof}
We will not state the full detail of the main theorem until Section
\ref{sect:maintheorem}, seeing that we need to first clarify notations
and definitions. Suffice it to say for now that under some technical
assumptions (a subset of that which was assumed in \cite{AlIoKl2010})
and a smallness condition (that the space-time is everywhere locally
close to Kerr-Newman), the event horizon of a stationary asymptotically 
flat solution to the  Einstein-Maxwell equations can have at most one 
connected component. 

We obtain the conclusion by studying a Cauchy hypersurface of the
domain of outer communications of such space-time. We show that its
topology must be that of $\Real^3$ with a single ball removed. We
argue by contradiction using a ``mountain pass lemma'' applied to the
function we denote by $y$, representing the real part of the inverse
of the Ernst potential. We will show
\begin{itemize}
\item Firstly, the function $y$ is well-defined in the domain of
outer communications. Noting that $y$ is defined by the inverse 
of the values of a smooth function, we need to show that the Ernst potential does not
vanish. This will occupy the bulk of the paper. 
\item Secondly, we need to show that $y$ satisfies the hypotheses of a
mountain pass lemma. To do so we use quantitative estimates derived
from the smallness conditions. On the domain of outer communications
of Kerr-Newman space-time, the function $y$ attains its minimum
precisely on the event horizon, and does not admit any critical points
outside the event horizon. We show that these properties remain
approximately true for our solutions. 
\item Lastly, to conclude the theorem, we observe that were there to
be more than one ``hole'' in the Cauchy hypersurface, the function $y$
must be ``small'' along two disconnected sets (the event horizons), 
and ``big'' somewhere away from those two sets. 
By the mountain pass lemma $y$ must then have a critical point, 
which gives rise to the contradiction. 
\end{itemize}

Our proof given in this manuscript is essentially perturbative. The
eventual goal, however, is to arrive at a ``large data'' theorem which
bypasses the smallness requirement in \ref{ass:kerrnewman}. At present
it is not clear to the authors how to proceed. While it is probable
that the eventual proof for the black hole uniqueness theorem does not
in fact make use of the characterisation tensors (see next section),
we hope the readers would forgive us for hoping that, given the
topological nature of the current argument, a ``large data'' version
of the presented theorem may be approachable if one were to find a
suitable geometric flow which acts ``monotonically'' (in a suitable
sense) on the characterisation tensors. 

\section{Preliminaries}
We begin with definitions. A space-time $(\spacetime,g_{ab})$ ---
that is, \begin{inparaenum}[(i)]
\item a four-dimensional, orientable, para-compact,
simply-connected manifold \spacetime~endowed with  
\item a Lorentzian metric $g_{ab}$ with signature $(-+++)$ such that
$(\spacetime,g_{ab})$ is time-orientable
\end{inparaenum} --- is said to be electro-vacuum if there exists a
(real) two-form $H_{ab}$ on \spacetime~called the \emph{Faraday tensor} such
that the Einstein-Maxwell-Maxwell (to distinguish it from non-linear
electromagnetic theories such as Einstein-Maxwell-Born-Infeld
\cite{Kiessl2004, Kiessl2004a, Speck2008}) equations are satisfied:
\begin{align*}
\Ricci_{ab} &= 2 H_{ac}H_b{}^c - \frac12 g_{ab}H_{cd}H^{cd} \\
& \left( = (H+
i{}^*H)_{ac}(H-i{}^*H)_b{}^c \right) \\
\nabla^a(H+i{}^*H)_{ac} &= 0
\end{align*}
where ${}^*$ is the Hodge-star operator: ${}^*H := \frac12
\lcsym_{abcd}H^{cd}$ with $\lcsym_{abcd}$ the volume form for the
metric $g_{ab}$. On a four-dimensional Lorentzian manifold, Hodge-star
defines an endomorphism on the space of two-forms which squares to
negative the identity. Hence we can factor over the complex numbers
and call a complex-valued two-form $\mathcal{X}_{ab}$ (anti-)self-dual if
${}^*\mathcal{X}_{ab} = (-)i\mathcal{X}_{ab}$. (See Section 2.1 in
\cite{Wong2009}
for a more detailed discussion of self-duality.) Observe that $H_{ab}
+ i{}^*H_{ab}$ is anti-self-dual. So equivalently we say the
space-time is electro-vacuum if there exists a complex, anti-self-dual
two-form $\maxw_{ab}$ such that 
\begin{subequations}
\begin{align}
\label{eq:einsteq}\Ricci_{ab} &= 4\maxw_{ac}\bar{\maxw}_b{}^c \\
\label{eq:maxeq}\nabla^a\maxw_{ac} &= 0~.
\end{align}
\end{subequations}
One can easily convert between the two formulations by the
formulae $2\maxw_{ab} = H_{ab} + i{}^*H_{ab}$, and $H_{ab} =
\maxw_{ab} + \bar{\maxw}_{ab}$. 

Throughout we will assume the electro-vacuum space-time
$(\spacetime,g_{ab},\maxw_{ab})$ admits a continuous symmetry, that
is, there exists a vector field $t^a$ on \spacetime~such that the Lie
derivatives $\lieD_tg_{ab} = 0$ ($t^a$ is Killing) and $\lieD_t\maxw_{ab} = 0$. 

We will use $C_{abcd}$ to denote the Weyl curvature, and
$\mathcal{C}_{abcd} = \frac12 (C_{abcd} + i{}^*C_{abcd})$ its
anti-self-dual part (see Section 2.2 of \cite{Wong2009}). For
an arbitrary tensor field $Z^{a_1\ldots a_k}_{b_1\ldots b_j}$ we write
$Z^2$ for its Lorentzian norm relative to the metric $g_{ab}$,
extended linearly to complex-valued fields. Hence for real $Z$, $Z^2$
may carry either sign; for complex $\mathcal{Z}$, $\mathcal{Z}^2$ can
be a complex number. We also define 
\[\Iasd_{abcd} := \frac14( g_{ac}g_{bd} - g_{ad}g_{bc} + i
\lcsym_{abcd}) \]
the projector to, and induced metric on, the space of anti-self-dual
two-forms. We also introduce the short-hand
\begin{equation}
(\mathcal{X}\tilde{\otimes}\mathcal{Y})_{abcd} :=
\frac12(\mathcal{X}_{ab}\mathcal{Y}_{cd} +
\mathcal{Y}_{ab}\mathcal{X}_{cd}) -
\frac13\Iasd_{abcd}\mathcal{X}_{ef}\mathcal{Y}^{ef}
\end{equation} 
which combines two anti-self-dual two-forms to form an anti-self-dual
Weyl-type tensor. 

Two important product properties of anti-self-dual two-forms that will be
used frequently in computations are
\begin{align}
\label{eq:asdprodprop1}\mathcal{X}_{ac}\bar{\mathcal{X}}_b{}^c &=
\mathcal{X}_{bc}\bar{\mathcal{X}}_a{}^c~,\\
\label{eq:asdprodprop2}\mathcal{X}_{ac}\mathcal{Y}_b{}^c + \mathcal{Y}_{ac}\mathcal{X}_b{}^c
&= \frac12 g_{ab} \mathcal{X}_{cd}\mathcal{Y}^{cd}~.
\end{align}

Lastly the symbols $\Re$ and $\Im$ will mean to take the real and
imaginary parts respectively. 

\subsection{The ``error'' tensors}
Now, since \maxw~solves Maxwell's equations, it is closed. Cartan's
formula gives
\[ d\iota_t\maxw + \iota_td\maxw = \lieD_t\maxw \]
and hence by our assumptions $\iota_t\maxw$ is a closed form. Since we
assumed our space-time is simply connected (a reasonable hypothesis in
view of the topological censorship theorem \cite{FrScWi1993} since we will only consider
a neighbourhood of the domain of outer communications), up to a
constant there exists some complex-valued function \maxpot~such that
$d\maxpot = \iota_t\maxw$. 

Observe that since $t^a$ is Killing, $\nabla_at_b$ is anti-symmetric.
Define $\preernst_{ab} = \nabla_at_b + 
\frac{i}{2}\lcsym_{abcd}\nabla^ct^d$. Now we define the complex
Ernst two-form 
\begin{equation}\label{eq:def:ernstform}
\ernstform_{ab} := \preernst_{ab} - 4 \bar\maxpot \maxw_{ab}~.
\end{equation}
One easily checks that \ernstform~also satisfies Maxwell's equations,
by virtue of the Jacobi equation for the Killing vector field $t^a$
(which is to say, $\covD_c\covD_at_b = \Riemann_{dcab}t^d$)
which implies that $\nabla^a\preernst_{ab} = - \Ricci_{ab}t^a$. Thus
analogous to how \maxpot~is defined, we can define (again up to a
constant) \ernstpot~to be a complex valued function, which we call the
Ernst potential, such that
$d\ernstpot = \iota_t\ernstform$. 

The main objects we consider are
\begin{defn}
The \emph{characterization} or \emph{error} tensors are the following
objects defined up to the free choice of four normalizing constants: the 
two complex 
constants in the definition of \ernstpot~and \maxpot, a complex
constant $\kappa$, and a real constant $\mu$. We define the two-form
\msmax~and the four-tensor \msweyl~by
\begin{subequations}
\begin{align}
\label{eq:def:msmax} \msmax_{ab} &:= \kappa \ernstform_{ab} + 2 \mu \maxw_{ab} \\
\label{eq:def:msweyl} \msweyl_{abcd} &:= \weyls_{abcd} + \frac{6\kappa \bar{\maxpot} - 3
\mu}{2\mu \ernstpot}(\ernstform \tilde\otimes \ernstform)_{abcd}
\end{align}
\end{subequations}
\end{defn} 

\begin{rmk}\label{rmk:afcompatibility}
The necessity of normalisation of $\maxpot$ and $\ernstpot$ is familiar
from classical physics: potential energies are relative and not
absolute. As it turns out, the choice of the four normalisation
constants entails compatibility with asymptotic flatness (that of the
potentials \maxpot\ and \ernstpot) and partial restrictions on the
mass, charge, and angular momentum parameters of the corresponding
Kerr-Newman solution. In the asymptotically flat case where space-like
infinity is defined and where we can read-off the asymptotic mass and
charge, the ``correct'' choice (see Assumption \ref{ass:kerrnewman}
below) of the normalising constants are such
that \maxpot\ and \ernstpot\ vanish at space-like infinity and $\mu$ 
and $\kappa$ are the mass and charge parameters respectively, as
these are the choices for which $\msmax$ and $\msweyl$ vanish in
Kerr-Newman space-time. When considering just a domain in space-time
when the asymptotically flat end is not accessible, we do not have a
canonical method of choosing the four constants; see also Theorem 
\ref{thm:character} below. 
\end{rmk}

These tensors are the natural generalization of the Mars-Simon tensor
\cite{Mars1999, Simon1984, IonKla2009a} which characterizes Kerr space-time among
stationary solutions of the Einstein vacuum equations. (Indeed, for
vacuum space-times we can set $\maxw$ and $\maxpot$ to be zero
identically; then choosing $\kappa = 0$ we have that $\msmax$ vanishes
and $\msweyl$ is exactly the Mars-Simon tensor.) More precisely,
we have the following theorem due to the first author \cite{Wong2009}.
\begin{thm}\label{thm:character}
Let $(\spacetime, g_{ab},\maxw_{ab})$ be an electro-vacuum space-time
admitting the symmetry $t^a$. Let $U\subset \spacetime$ be a connected
open subset, and suppose there exists a normalisation such that on $U$
we have $\ernstpot\neq 0$, $\msmax = 0$, and $\msweyl = 0$. Then we
have 
\[ t^2 + 2\Re{\ernstpot} + \frac{|\kappa\ernstpot|^2}{\mu^2} + 1 =
\mathrm{const.} \qquad \mbox{and} \qquad \mu^2\ernstform^2 + 4\ernstpot^4 =
\mathrm{const.} \]
If, furthermore, both the above expressions evaluate to 0, and $t^a$
is time-like somewhere on $U$, then $U$ 
is locally isometric to a domain in Kerr-Newman
space-time with charge $\kappa$, mass $\mu$, and angular momentum
$\mu\sqrt{\mathfrak{A}}$, where
\[ \mathfrak{A} := \left|\frac{\mu}{\ernstpot}\right|^2\left(
\Im\nabla\frac{1}{\ernstpot}\right)^2 +
\left(\Im\frac{1}{\ernstpot}\right)^2 \]
is a constant on $U$. 
\end{thm}
\begin{rmk}
Algebraically the definitions given herein are normalized differently
from the definitions in \cite{Wong2009}. For $\kappa \neq 0$ by
rescaling one can see that the statements in the above theorem are
algebraically identical to the hypotheses in the main theorem in
\cite{Wong2009}. For $\kappa = 0$ it is trivial to check that the
conditions given above reduces to the case given in \cite{Mars1999}. 
\end{rmk}
\begin{rmk}
The condition that $t^a$ is time-like somewhere on $U$ can be relaxed to the
condition that there is some point in $U$ where $t^a$ is not
orthogonal to either of the principal null directions of $\ernstform$. 
Also note that asymptotic flatness is not required for the theorem; in
the asymptotically flat case, using the normalisation described in
Remark \ref{rmk:afcompatibility}, the vanishing of \msmax\ and
\msweyl\ automatically ensures that the expressions involving
$t^2\ldots$ and $\mu^2 \ernstform^2\ldots$ vanishes. 
\end{rmk}

In view of Theorem \ref{thm:character}, we expect to use the tensors
\msmax~and \msweyl~as a measure of deviation of an arbitrary
stationary electro-vacuum solution from the Kerr-Newman family.
Indeed, the main assumption to be introduced in the next section is a
uniform smallness condition on the two tensors. In fact, we say that
\begin{defn}
A tensor $\mathcal{X}_{a_1\ldots a_k}$ is said to be an
\emph{algebraic error term} if there exists smooth tensors
$\mathcal{A}^{(1)}_{a_1\ldots a_k}{}^{bc}$,
$\mathcal{A}^{(2)}_{a_1\ldots a_k}{}^{bcd}$, and 
$\mathcal{A}^{(3)}_{a_1\ldots a_k}{}^{bcde}$ such that
\[ \mathcal{X}_{a_1\ldots a_k} = \mathcal{A}^{(1)}_{a_1\ldots
a_k}{}^{bc}\msmax_{bc} + \mathcal{A}^{(2)}_{a_1\ldots
a_k}{}^{bcd}\nabla_b\msmax_{cd} + \mathcal{A}^{(3)}_{a_1\ldots
a_k}{}^{bcde}\msweyl_{bcde}~.\]
\end{defn}
\begin{rmk}
In the course of the proof, we shall see explicit expressions for all
the algebraic error terms that play a role in the analysis. For these
error terms, the tensors $\mathcal{A}^{(*)}_*$ can be controlled by
the background geometry. See Assumption \ref{ass:kerrnewman} below as
well as Proposition \ref{prop:errest}.
\end{rmk}
Morally speaking, an algebraic error term is one that can be
``made small'' by putting suitable smallness assumptions on the error
tensors. In view of the indefiniteness of the Lorentzian geometric, the 
smallness needs to be stronger than smallness in ``Lorentzian norm'';
see Assumption \ref{ass:kerrnewman} in the next section. Of course, 
we note that should the black hole uniqueness theorem
be proved in the smooth category (as opposed to the state-of-the-art
that only holds for real-analytic space-times), then with some 
reasonable conditions imposed on the space-time \msmax~and \msweyl~ must 
vanish identically. 

Following the definition by Equation \eqref{eq:def:msmax}, we
immediately have 
\begin{lem}\label{lem:sumpotsisweakerror}
The exterior derivative $d\msmaxpot$ of the potential sum $\msmaxpot:=\kappa\ernstpot +
2\mu\maxpot$ is an error term. 
\end{lem}

For conciseness, we will also use the notation 
\begin{equation}\label{eqn:defnpzero}\maxpotr := 2\bar\kappa
\maxpot - \mu~,
\end{equation} 
and define the real-valued
quantities $y,z$ such that 
\[ y+iz := -\ernstpot^{-1}\]
when the
right-hand side is finite. For motivation, we mention the main lemma
used in proving Theorem \ref{thm:character}. 
\begin{lem}[Mars-type Lemma \cite{Wong2009}]\label{lem:exactmars}
Under the assumptions of Theorem \ref{thm:character} with the
requirement that the two expressions evaluate to 0, we have
$g^{ab}\nabla_a y \nabla_b z = 0$ and
\[ (\nabla z)^2 = \frac{1}{\mu^2}\frac{\mathfrak{A} - z^2}{y^2 + z^2} \qquad (\nabla
y)^2 = \frac{1}{\mu^2}\frac{\mathfrak{A} + |\kappa/\mu|^2 + y^2 -
2y}{y^2 + z^2} \]
for the constant $\mathfrak{A}$ as given in Theorem
\ref{thm:character}. 
\end{lem}

Compare the above lemma to Lemma \ref{lem:mainlemma} which gives the
analogous statement under the condition \msmax\ and \msweyl\ are
small, but not necessarily vanishing. For the
expression involving $(\nabla z)^2$, we note that $\mathfrak{A}$ is
now no longer a constant, but \emph{almost} so. For the expression
involving $(\nabla y)^2$, we apply \eqref{eq:identityc2} of Corollary
\ref{cor:identities} and pick up a few additional error terms. For the 
statement about orthogonality of $\nabla y$ and $\nabla z$, see 
\eqref{eq:identityc1} of Corollary \ref{cor:identities}. 

\subsection{Geometric assumptions and the Main
Theorem}\label{sect:maintheorem}
Now we provide the precise set-up for our main theorem. 
\newcounter{asscount}
\renewcommand{\theasscount}{\textbf{\small(\ifcase\value{asscount}\or
TOP\or AF\or SBS\or KN\else nomore\fi\relax)}} 
\begin{list}{\theasscount}{\usecounter{asscount}\setlength{\parsep}{0.4em}\setlength{\itemsep}{0.7em}}
\item \label{ass:topology} We assume that there is a embedded partial Cauchy
hypersurface $\cauchy\subset\spacetime$ which is space-like everywhere. To
model the multiple black holes we assume, in view of the Topology
Theorem \cite{GalSch2006}, that \cauchy~is diffeomorphic to
$\Real^3\setminus \cup_{i = 1}^\mathfrak{k} B_i$, which is the Euclidean
three-space with finitely many disjoint balls removed. We denote the
diffeomorphism by 
\[ \Phi: \Real^3\setminus \cup_{i = 1}^\mathfrak{k} B_i \to \cauchy \]
and require that $\mathfrak{k}$ is the total number of black holes. Each
$B_i$ is a ball centered at $b_i$ with radius $\frac12$. We also
require that $|b_i - b_j| > 3$ when $i\neq j$. On
$\Real^3$ we use the usual Euclidean coordinate functions
$(x^1,x^2,x^3)$ with the convention $r = \sqrt{ (x^1)^2 + (x^2)^2 +
(x^3)^2}$. Thus for large enough $R_0$ the set $E(R_0) := \{p \in
\Real^3 \setminus \cup_{i = 1}^\mathfrak{k} B_i | r > R_0\}$ is unambiguously
$\Real^3$ with a large ball removed. 

Furthermore we assume that for sufficiently large $R_0$, the Killing
vector field $t^a$ is transverse to $E(R_0)$, and thus by integrating
along the symmetry orbits we extend a diffeomorphism
\[ \tilde{\Phi}: \Real \times E(R_0) \to \spacetime^{\mathrm{end}} \]
where $\spacetime^{\mathrm{end}}$ is an open subset in
\spacetime~which we call the \emph{asymptotic region}. In
particular this defines local coordinates $(x^0,x^1,x^2,x^3)$ on
$\spacetime^{\mathrm{end}}$ with $t = \partial_0$.  
\item\label{ass:asympflat} In view of the dipole expansions in \cite{MiThWh1973}
(see also \cite{BeiSim1981}), we assume the following asymptotic
properties for the metric and Faraday tensors in the local coordinates
on $\spacetime^{\mathrm{end}}$. The notation $O_k(r^m)$
stands for smooth functions $f$ obeying $|\partial^\beta f| \lesssim
r^{m - |\beta|}$ for any multi-index $\beta$ with $0 \leq |\beta |
\leq k$. The metric components are
\begin{equation}
\left\{ \begin{array}{rl}
g_{(0)(0)} &= -1 + 2Mr^{-1} + O_4(r^{-2}) \\
g_{(0)(i)} &= -2\sum_{j,k=1}^3\lcsym_{ijk}S^jx^kr^{-3} + O_4(r^{-3})\\
g_{(i)(j)} &= (1 + 2Mr^{-1})\delta_{ij} + O_4(r^{-2})
\end{array}\right.
\end{equation} 
where $(S^1,S^2,S^3)$ form the angular momentum vector and
$\lcsym_{ijk}$ is the fully anti-symmetric Levi-Civita symbol with 3
indices. $M > 0$ is, of
course, the ADM mass. Using the gauge
symmetry of the Maxwell-Maxwell equations, we shall apply a charge
conjugation and assume that the space-time carries a total electric
charge $q \geq 0$ and no magnetic charge. Then components
of the Faraday tensor read
\begin{equation}
\left\{ \begin{array}{rl}
H_{(i)(0)} &= \dfrac{q}{r^3} x^i + O_4(r^{-3}) \\
H_{(i)(j)} &= \dfrac{q}{Mr^{3}}\sum_{k =
1}^3\lcsym_{ijk}\left(\frac{3\sum_{l=1}^3S^lx^l}{r^2} x^k
- S^k\right)
+ O_4(r^{-4})
\end{array}\right.
\end{equation}
We define the total angular momentum of the space-time to be 
\begin{equation}
\mathfrak{a}^2 := \frac{(S^1)^2 + (S^2)^2 + (S^3)^2}{M^2}
\end{equation}
and require the non-extremal condition 
\begin{equation}
q^2 + \mathfrak{a}^2 < M^2
\end{equation}
to hold. 
\item\label{ass:smoothbifsphere} Define $\outercom :=
\mathcal{I}^-(\spacetime^{\mathrm{end}})\cap
\mathcal{I}^+(\spacetime^{\mathrm{end}})$ to be the domain of outer
communications. We assume that $\outercom$ is globally hyperbolic
and 
\begin{equation}
\cauchy \cap \mathcal{I}^-(\spacetime^{\mathrm{end}}) = \cauchy \cap
\mathcal{I}^+(\spacetime^{\mathrm{end}}) = \Phi(\Real^3 \setminus
\cup_{i = 1}^\mathfrak{k} B'_i)
\end{equation}
where $B'_i$ are balls of radius 1 centered at $b_i$, i.e.~they are
concentric with the balls $B_i$ but have twice the radii. Furthermore, we require that 
\[ \Phi(\cup_{i
= 1}^{\mathfrak{k}} \partial B'_i) =
\partial\mathcal{I}^-(\spacetime^{\mathrm{end}})\cap
\partial\mathcal{I}^+(\spacetime^{\mathrm{end}}) \]
in other words, that $\cauchy$ passes through the bifurcate spheres of
all black holes. (Physically this suggests that the black holes are 
``space-like'' relative to each other.) Note that our choice of
coordinates in \ref{ass:topology} implies that the bifurcate spheres
are pairwise at least coordinate-distance 1 apart. We denote by $\horizon^0_i
= \Phi(\partial B'_i)$. Write $\horizon^+ =
\partial\mathcal{I}^-(\spacetime^{\mathrm{end}})$ and $\horizon^- =
\partial\mathcal{I}^+(\spacetime^{\mathrm{end}})$; let $\horizon^0 =
\cup_{i = 1}^k \horizon^0_i$, and denote by $\horizon^\pm_i$ the
connected component of $\horizon^\pm$ containing $\horizon^0_i$. We
shall assume each $\horizon^\pm_i$ is a smooth, embedded, null
hypersurface, and require that $\horizon^+_i$ and $\horizon^-_i$
intersects transversely at $\horizon^0_i$. We remark that the
existence of $t^a$ ensures that each $\horizon^\pm_i$ is
non-expanding by Hawking's Area Theorem (see, e.g.~\cite[Theorem
7.1]{CDGH2001}), i.e.~has vanishing null second fundamental form, and
that $t^a$ is tangent to each $\horizon^\pm_i$ (see \cite[Chapter
2]{Wong2009a}
for more detailed discussion of these facts). 
We assume that the orbits of $t^a$ are
complete in $\outercom$ and are transverse to
$\outercom\cap\cauchy$. 
\item\label{ass:kerrnewman} Under the asymptotic flatness, we shall fix $\maxpot$
and $\ernstpot$ by requiring that they asymptotically vanish as
$r\nearrow +\infty$, and we set $\mu = M$ and
$\kappa = q$ in the definition of \msweyl~ and \msmax. 
Fix, once and for all, a coordinate system
$(x^0,x^1,x^2,x^3)$ in a
tubular neighbourhood of \cauchy~such that it agrees with the
coordinate system at $\spacetime^{\mathrm{end}}$ (perhaps after
enlarging $R_0$) and such that the $x^1, x^2, x^3$ functions when
restricted to $\cauchy$ agrees with that induced by $\Phi$. We require
that the
metric $g$, its inverse, its Christoffel symbols, the Faraday
tensor $H$ and the Killing vector field $t^a$ are uniformly bounded 
in the coordinates. We then impose the following smallness assumption
along \cauchy: for some $\epsilon$ sufficiently small (depending only
on
$M,q,\mathfrak{a}$, the number $R_0$, and the uniform bound above) we
have the bound 
\begin{equation}\label{eqn:smallnessassumption}
\sum_{0\leq \alpha,\beta,\gamma,\delta\leq 3}|
\msweyl_{(\alpha)(\beta)(\gamma)(\delta)}| + \sum_{0\leq
\alpha,\beta\leq 3}| \msmax_{(\alpha)(\beta)} | + \sum_{0\leq
\alpha,\beta,\gamma\leq 3}
|\partial_{(\gamma)}\msmax_{(\alpha)(\beta)}| < \epsilon |\maxpotr|
\end{equation}
(recall that $\maxpotr$ is defined by \eqref{eqn:defnpzero}) 
where $\partial$ denotes coordinate derivative, and $(a)$ denotes
coordinate evaluation of the tensor object. 
\end{list}

Our main theorem is 
\begin{thm}[Non-existence of multi-black-holes]\label{thm:mainthm}
Under the assumptions \ref{ass:topology}, \ref{ass:asympflat},
\ref{ass:smoothbifsphere}, and \ref{ass:kerrnewman},
$\mathfrak{k}$ (the number of components of the horizon) must equal 1. In other
words, there can only be one black hole.
\end{thm}
\begin{rmk}
Under the above definitions, we can recover the Einstein-vacuum case
directly as a corollary. Note that by \emph{a priori} setting, in the 
hypotheses to Theorem \ref{thm:mainthm}, $q = 0$ and taking the
Faraday tensor $H_{ab}\equiv 0$, we restrict ourselves to
stationary Einstein-vacuum solutions with only vacuum perturbations. 
\end{rmk}

\begin{rmk}
We should compare the smallness condition
\eqref{eqn:smallnessassumption} to that given in \cite{AlIoKl2010}.
The contribution from the electromagnetic field requires us to
introduce the term $\maxpotr$ on the right hand side. In the pure
vacuum case, we see from its definition that $\maxpotr = -\mu = - M <
0$ and can be absorbed into $\epsilon$. If we compute $\maxpotr$ using
the explicit Kerr-Newman metric, we see that in the exterior region
\outercom, we have that $|\maxpotr| > M^2 - q^2$ (where the
minimum is achieved at the ``poles'' of the bifurcate sphere) and is 
bounded away from zero uniformly for subextremal parameters. Hence 
for \emph{bona fide} small
perturbations (in the sense that we are given a fixed coordinate
system and in this coordinate system the metric $g$, its inverse, the
Killing field $t$, and the Faraday tensor $H$ are all uniformly $C^2$
close to that of a background Kerr-Newman solution) the right-hand-side of
\eqref{eqn:smallnessassumption} can be replaced by a fixed constant. 
The factor of $\maxpotr$ is needed to control some error
terms in the case of a hypothetical electro-vacuum space-time that is
not a bona fide small perturbation in the sense above, yet still
has suitably small error tensors; see Proposition \ref{prop:errest} and Lemma
\ref{lem:unoderiserror}. 

In addition, \eqref{eqn:smallnessassumption} seemingly requires one
more derivative compared to the condition assumed in
\cite{AlIoKl2010}. However, observe that in the vacuum case we can
choose $\kappa = 0$ and set $\msmax_{ab}\equiv 0$ automatically.
In that case our $\msweyl_{abcd}$ agrees with the vacuum Mars-Simon
tensor, and there is no derivative loss when restricted to the special
case. That matter fields are ``one derivative worse'' than the metric
is a recurring theme in mathematical relativity, see, in a different
context, \cite{Shao2011}. 
\end{rmk}

\subsection{Algebraic lemmas}
In this section we document some algebraic manipulations that will be
useful in the sequel. Note that unless specified, none of the four
assumptions \ref{ass:topology}, \ref{ass:asympflat},
\ref{ass:smoothbifsphere}, and \ref{ass:kerrnewman} are used. The
identities we derive, of course, will only hold when both sides of the
equal sign are well-defined. Part of the bootstrap in the proof of the
main theorem shall be demonstrating that all the quantities in these
identities remain finite and smooth.

First we note some immediate consequences of Equation
\eqref{eq:def:msmax} that measure the differences between $\preernst$,
$\ernstform$, and $\maxw$ in terms of $\msmax$:
\begin{subequations}
\begin{align}
\label{eq:preernsternst}2\bar\maxpot \msmax_{ab} - \mu \preernst_{ab} &= \bar\maxpotr \ernstform_{ab}\\
\label{eq:preernstmaxw}\kappa\preernst_{ab} - \msmax_{ab} &= 2\bar\maxpotr \maxw_{ab}
\end{align}
\end{subequations}
Hence
\begin{align*}
\bar\maxpotr \nabla_c\ernstform_{ab} &= 2 \nabla_c(\bar\maxpot
\msmax_{ab} ) - \mu \nabla_c\preernst_{ab} - \nabla_c\bar\maxpotr
\ernstform_{ab} \\
&= 2 \nabla_c(\bar\maxpot \msmax_{ab} ) - 2\kappa \nabla_c\bar\maxpot
\ernstform_{ab} - 2\mu \weyls_{dcab}t^d  \\
& \qquad - 2\mu (\Ricci_d{}^eg_c{}^f -
\Ricci_c{}^eg_d{}^f)\Iasd_{efab}t^d
\end{align*}
via the Jacobi equation for the Killing vector field $t^a$ (see, e.g.\
equation (C.3.6) of \cite{Wald1984}). Thus
\begin{align*}
\frac12 \bar\maxpotr \nabla_c\ernstform^2 &= 2 \ernstform^{ab}
\nabla_c(\bar\maxpot \msmax_{ab} ) - 2\kappa \ernstform^2
\nabla_c\bar\maxpot - 2\mu \msweyl_{dcab}\ernstform^{ab}t^d \\
& \qquad + \frac{3\bar\maxpotr}{\ernstpot}
(\ernstform\tilde\otimes\ernstform)_{dcab}\ernstform^{ab} t^d
- 2\mu (\Ricci_{de}g_{cf}- \Ricci_{ce}g_{df})\ernstform^{ef}t^d \\
& = 2 \ernstform^{ab} \nabla_c(\bar\maxpot \msmax_{ab} ) - 2\mu
\msweyl_{dcab}\ernstform^{ab}t^d  - 2\kappa \ernstform^2
\bar\maxw_{dc}t^d + \frac{2\bar\maxpotr}{\ernstpot}
\ernstform^2\ernstform_{dc}t^d \\
& \qquad - 4 [\bar\maxw_{da}(\msmax^{ea} - \kappa \ernstform^{ea})\ernstform_{ec} -
\bar\maxw_{ca}(\msmax^{ea} - \kappa\ernstform^{ea})\ernstform_{ed}]t^d \\
& = 2 \ernstform^{ab} \nabla_c(\bar\maxpot \msmax_{ab} ) - 2\mu
\msweyl_{dcab}\ernstform^{ab}t^d - 4(\bar\maxw_{da}\ernstform_{ec} -
\bar\maxw_{ca}\ernstform_{ed})\msmax^{ea}t^d \\
& \qquad - 2\kappa \ernstform^2
\bar\maxw_{dc}t^d + \frac{2\bar\maxpotr}{\ernstpot}
\ernstform^2\ernstform_{dc}t^d + 2\kappa \bar\maxw_{dc} \ernstform^2 t^d
\end{align*}
From which we conclude
\begin{align}
\bar\maxpotr \ernstpot^4
& \nabla_c\left(\frac{\ernstform^2}{4\ernstpot^4}\right) \\
\nonumber & =
\ernstform^{ab}\left[ \nabla_c(\bar\maxpot\msmax_{ab}) - \mu
\msweyl_{dcab}t^d\right] - 2(\bar\maxw_{da}\ernstform_{ec} -
\bar\maxw_{ca}\ernstform_{ed})\msmax^{ea}t^d
\end{align}
In other words
\begin{lem}\label{lem:unoderiserror}
The quantity $\bar\maxpotr \ernstpot^4
\nabla_c(\ernstform^2/4\ernstpot^4)$ is an algebraic error term.
\end{lem}

Next we show that
\begin{lem}
The following identities hold: 
\begin{align}
\label{eq:identity1}\left( \nabla \frac{1}{\ernstpot} \right)^2 & =
\frac{\ernstform^2}{4\ernstpot^4} t^2 \\
\label{eq:identity2} -|\kappa|^2 t^2 &= \Re(
2\bar\kappa \msmaxpot) + |\maxpotr|^2 + \mathrm{const.}\\
\intertext{also}
\label{eq:identity2p}\tag{\ref*{eq:identity2}$'$} -t^2 - 1 &=
\frac{1}{\mu^2}\left| \msmaxpot - \kappa \ernstpot\right|^2 + \ernstpot
+ \bar\ernstpot + \mathrm{const.} \\
\intertext{and lastly}
\label{eq:identity3} \Box \frac{1}{\ernstpot} &= -
\frac{\ernstform^2}{2\ernstpot^3}(1+ \mathrm{const.} + \bar\ernstpot) \\
\nonumber & \qquad + \frac{\bar\maxpot}{\mu\ernstpot^2} \ernstform\cdot\msmax -
\frac{1}{\mu^2}\frac{\ernstform^2}{\ernstpot^3}\msmaxpot
\overline{(\msmaxpot - \kappa\ernstpot)}
\end{align}
where the constants in \eqref{eq:identity3} and \eqref{eq:identity2p}
are the same.
\end{lem}
\begin{rmk}\label{rmk:AFonconstants}
Under the asymptotic flatness assumption \ref{ass:asympflat}, our
normalization convention fixes $\maxpot$ and $\ernstpot$ to vanish at
spatial infinity; by definition $\msmaxpot$ also tends to zero, while
$\maxpotr$ tends to $-\mu$. Hence under this assumption, the free
constant in \eqref{eq:identity2} will be $|\kappa|^2-\mu^2$, and the
constants in \eqref{eq:identity2p} and \eqref{eq:identity3} will both
be 0.
\end{rmk}
\begin{proof}
The first equation \eqref{eq:identity1} can be directly derived by
appealing to the definitions: noting that $\nabla_a \ernstpot =
\ernstform_{ba}t^a$, we obtain $(\nabla \ernstpot)^2 = \frac14
\ernstform^2 t^2$ by \eqref{eq:asdprodprop2}. The second expression follows from
\[
\nabla_a t^2 = 2 t^b\Re \preernst_{ab} = - 2 \Re \left( \frac{2}{\kappa}\bar\maxpotr\nabla_a\maxpot +
\frac{1}{\kappa}\nabla_a \msmaxpot \right) =
-\frac{1}{\kappa\bar\kappa} \nabla_a|\maxpotr|^2 + \nabla_a \Re
(\frac{2}{\kappa}\msmaxpot)~.
\]
The computation for \eqref{eq:identity2p} is slightly less trivial:
\begin{align*}
\nabla_a t^2 &= 2 t^b\Re \preernst_{ab} = -\frac{2}{\mu} \Re \left( 2\bar\maxpot \nabla_a\msmaxpot -
\bar\maxpotr \nabla_a\ernstpot\right) \\
& = -\frac{2}{\mu} \Re \left( 2\bar\maxpot \nabla_a (\msmaxpot -
\kappa\ernstpot) + \mu \nabla_a\ernstpot\right) \\
& = - 2 \Re \left( \frac{1}{\mu^2} \overline{(\msmaxpot -
\kappa \ernstpot)}\nabla_a(\msmaxpot - \kappa\ernstpot) +
\nabla_a\ernstpot\right) 
\end{align*}
And lastly we observe
\begin{align*}
\Box \frac1\ernstpot &= \nabla^a\nabla_a \frac1\ernstpot = -
\nabla^a\left( \frac{1}{\ernstpot^2} \ernstform_{ba} t^b\right) \\
& = -\frac{1}{\ernstpot^2}\ernstform_{ba}\nabla^at^b +
\frac{2}{\ernstpot^3}\ernstform_{ba}\ernstform^{ca}t^bt_c \\
& =\frac{1}{2\ernstpot^2}\ernstform_{ba}\preernst^{ba} +
\frac{1}{\ernstpot^3}\ernstform^2t^2 \\
& = \frac{\bar\maxpot}{\mu\ernstpot^2} \ernstform_{ba}\msmax^{ba} +
\frac{\ernstform^2}{2\ernstpot^3}\left( t^2 - \frac1\mu \ernstpot
\bar\maxpotr\right) \\
& = \frac{\bar\maxpot}{\mu\ernstpot^2} \ernstform_{ba}\msmax^{ba} +
\frac{\ernstform^2}{2\ernstpot^3}\left( t^2 + \ernstpot -
\frac{2\kappa\ernstpot}{\mu^2} \overline{(\msmaxpot -
\kappa\ernstpot)}\right) 
\end{align*}
Applying \eqref{eq:identity2p} we see
\begin{align*}
t^2 + \ernstpot - \frac{2\kappa\ernstpot}{\mu^2} \overline{(\msmaxpot -
\kappa\ernstpot)} & = t^2 + \ernstpot + \frac{2}{\mu^2} \left| \msmaxpot
-\kappa\ernstpot\right|^2 - \frac{2}{\mu^2}\msmaxpot
\overline{(\msmaxpot - \kappa\ernstpot)}\\
&= -\left(1 +\mathrm{const.} + \bar\ernstpot + \frac{2}{\mu^2}\msmaxpot
\overline{(\msmaxpot - \kappa\ernstpot)}\right)
\end{align*}
Combining the expressions we obtain \eqref{eq:identity3} as claimed.
\end{proof}
In view of Remark \ref{rmk:AFonconstants}, and recalling the
definition $(y+iz)^{-1} = -\ernstpot$ we have the following
expressions
\begin{cor}\label{cor:identities}
Under the asymptotic flatness assumption \ref{ass:asympflat}, 
\begin{subequations}
\begin{align}
\label{eq:identityc1}g^{ab}\nabla_a y \nabla_b z & = \frac{t^2}{2}
\Im\errsr{uno} \\
\label{eq:identityc2}(\nabla y)^2 - (\nabla z)^2 &= \frac{y^2 + z^2 - 2y +
\frac{|\kappa|^2}{\mu^2}}{\mu^2(y^2 + z^2)} + \frac{|\msmaxpot|^2 -
2\Re(\kappa\ernstpot\bar\msmaxpot)}{\mu^4} + t^2\Re\errsr{uno}\\
\label{eq:idenitityc3}\Box y + \frac{2}{\mu^2}\frac{1 - y}{y^2 + z^2} &= 2\Re\left(\ernstpot(1+\bar\ernstpot)
\errsr{uno} + \frac{1}{\ernstpot^2}\errsr{dos} -
8\ernstpot\bar\maxpot\errsr{tres}\right)\\
\label{eq:idenitityc4}\Box z + \frac{2}{\mu^2}\frac{z}{y^2 + z^2} &=
2\Im\left(\ernstpot(1+\bar\ernstpot)
\errsr{uno} + \frac{1}{\ernstpot^2}\errsr{dos} -
8\ernstpot\bar\maxpot\errsr{tres}\right)
\end{align}
\end{subequations}
Where the terms $\errsr{uno}, \errsr{dos}, \errsr{tres}$ are given by
\begin{align}
\errsl{uno} &= \frac{1}{\mu^2} + \frac{\ernstform^2}{4\ernstpot^4} &
\errsl{dos} &= \frac{1}{\mu}\bar\maxpot \ernstform\cdot\msmax &
\errsl{tres} &= \frac{1}{\mu}\frac{\ernstform^2}{4\ernstpot^4}\msmaxpot
\end{align}
each has the property that its exterior derivative is an algebraic
error term up to a multiplicative factor of $\ernstpot^{-4}$. 
\end{cor}

The following lemma is a refinement of a proposition first due to Mars
in the vacuum case \cite{Mars1999} (see also Lemma 10 in
\cite{Wong2009} for a version in charged space-times). In order to 
capture the exact contributions from the error tensors, we forgo the
tetrad formalisms used by Mars and by the first author in their
papers, and instead work directly and covariantly with the tensors,
improving upon the approach taken by Alexakis, Ionescu, and Klainerman 
\cite{AlIoKl2010}. As a consequence, the proof is lengthy, and
we defer its presentation to Appendix \ref{appendix1}.

\begin{lem}[Main lemma]\label{lem:mainlemma}
Define the quantity $\mathfrak{A} := \mu^2(y^2+z^2)(\nabla z)^2 + z^2$, 
then $\mathfrak{A}$ is ``almost constant''. More precisely,
\begin{align}
\notag \nabla_a\mathfrak{A} &=
\frac{4\mu^2}{|\ernstpot|^2}\nabla^bz\Im\left(
\frac{t^c}{\ernstpot^2\bar\maxpotr}\nabla_a\msmax_{cb} -
\frac{\mu}{\ernstpot^2\bar\maxpotr}\msweyl_{dacb}t^ct^d\right)  + 2
\nabla_az
\Im\left(\frac{\bar\kappa\bar\ernstpot}{\mu^2\ernstpot}\msmaxpot\right)
\\
& \qquad + \mu^2 t^2(z\nabla_a y - y\nabla_a z)\Im\errsr{uno} -
\Im\left[\frac{2\errsr{uno}\mu^2}{|\ernstpot|^2}\nabla_az\left(\ernstpot
t^2 + \frac{i}{\mu}\Im(\bar\ernstpot^2\maxpotr)\right)\right] \\
\notag & \qquad - \frac{z\nabla_az}{\mu^2}(|\msmaxpot -
\kappa\ernstpot|^2 -
|\kappa\ernstpot|^2)  +
\Im\left[\frac{4\mu^3}{|\ernstpot|^2\ernstpot^2\bar\maxpotr}\nabla^bz
(\errsr{appone})_{ab}\right] \\
\notag & \qquad +
\Im\left[\frac{4\mu}{|\ernstpot|^2\ernstpot^2}\ernstform_{cb}\Re(\bar\maxpot\msmax_a{}^c)\nabla^bz
-
\frac{\mu\maxpotr\bar\ernstpot}{\ernstpot}\Im(\errsr{uno})\nabla_a\ernstpot^{-1}\right]
\end{align}
where $\errsr{appone}$ is defined in \eqref{eq:deferrorappone} in the
appendix. Each term on the right hand side either contains an
algebraic error term, or contains a factor of $\msmaxpot$ or
$\errsr{uno}$, whose derivatives are algebraic error terms. 
\end{lem}

\subsection{Null decomposition}\label{sec:nulldec}

In regions where $\ernstform^2 \neq 0$, the Ernst two-form is
non-degenerate and anti-self-dual, and has two distinct, future
directed, principal null directions $l^a$ and $\lbar^a$, which we will
normalize to $g_{ab}l^a\lbar^b = -1$. So there exists a
complex-valued scalar function $\ernstnd$ such that
\[ \ernstform_{ab} = \ernstnd\left( \lbar_a l_b - l_a\lbar_b +
i\lcsym_{abcd}\lbar^c l^d\right)~.\]
Immediately we have $\ernstform^2 = -4\ernstnd^2$. 

We can then decompose $\nabla_a y$ and $\nabla_a z$ by noting that
$\nabla_a (-\ernstpot^{-1}) = \ernstpot^{-2} \ernstform_{ba} t^b$.
\begin{subequations}
\begin{align}
\nabla_a y &= \pm\frac{1}{\mu}\left(\lbar\cdot t l_a - l\cdot t
\lbar_a\right) + \Re\left[\errsr{cuatro} \left( \lbar\cdot t l_a - l\cdot
t\lbar_a + i \lcsym_{bacd}t^b\lbar^cl^d\right)\right]
\\
\nabla_a z &= \pm\frac{1}{\mu}\lcsym_{bacd} t^b\lbar^cl^d +
\Im\left[\errsr{cuatro} \left( \lbar\cdot t l_a - l\cdot
t\lbar_a + i \lcsym_{bacd}t^b\lbar^cl^d\right)\right]
\\
\errsl{cuatro} &= 
\left(\frac{\ernstnd}{\ernstpot^2}\mp\frac{1}{\mu}\right)
\end{align}
\end{subequations}
The $\pm$ signs in the above signal two equivalent local definitions.
We will always make use of the one with the smaller
$|\errsr{cuatro}|$; with this choice, we can estimate $\errsr{cuatro}$
by $\errsr{uno}$.  Indeed, $(\ernstnd/\ernstpot^2 -
1/\mu)(\ernstnd/\ernstpot^2 + 1/\mu) = - \ernstform^2/4\ernstpot^2 -
1/\mu^2 = -\errsr{uno}$. So $\errsr{cuatro}$ satisfies an equation of
the form 
\[ \big| \errsr{cuatro}\big|~ \big|\errsr{cuatro} \mp
2/\mu\big| = \big|\errsr{uno}\big|~.\]
By assumption that $|\errsr{cuatro}| \leq |\errsr{cuatro} \mp 2/\mu|$
with the appropriate sign, hence we have that $|\errsr{cuatro}| <
\sqrt{|\errsr{uno}|}$. Now, if $\mu \geq 1 / \sqrt{|\errsr{uno}|}$, we
have that $|\errsr{cuatro}| \leq \mu |\errsr{uno}|$. On the other
hand, if $\frac{1}{\mu} \geq \sqrt{|\errsr{uno}|}$, we have that
\[ 
|\errsr{cuatro}| \leq \frac{1}{\mu} \implies \left|\errsr{cuatro} \mp
\frac2\mu\right| \geq \frac{1}{\mu} 
\]
by the triangle inequality. And so in either case we can conclude
\begin{equation} |\errsr{cuatro}| \leq \mu|\errsr{uno}|~.
\end{equation}
This in particular implies that up to an error controlled by
$\errsr{cuatro}$, the gradient $\nabla_az$ is
space-like, which will imply, via Lemma \ref{lem:mainlemma}, that $z$
is almost bounded.

\section{Domain of definition of the function $y$}
\setcounter{equation}{0}
The first step in the proof of Theorem \ref{thm:mainthm} is to
establish that the function $y$ is well-defined and smooth to the
exterior of the black hole. More precisely, we claim that
\begin{prop}\label{prop:domdefy}
Under the hypotheses of Theorem \ref{thm:mainthm}, where the constant
$\epsilon$ in assumption \ref{ass:kerrnewman} is taken to be
appropriately small, the function $\ernstpot$ does not vanish on
$\bar\outercom$, the closure of the domain of outer communication. In
particular, this implies that $y$ is smooth on $\outercom$ and extends
continuously to $\bar\outercom$. 
\end{prop}

We devote the current section to the proof of the above proposition.
This proposition is an analogue of Proposition 3.4 in
\cite{AlIoKl2010}. While the basic ideas for the proof via a
``bootstrap'' argument from infinity is the same, because of the more
complicated forms of error terms coming from the electromagnetic
coupling, we choose to give a different presentation to make clear the
roles played by the various algebraic error terms. In particular, it
is necessary in our analysis that the right hand side of
\eqref{eqn:smallnessassumption} contains $\maxpotr$ which could \emph{a
priori} vanish. In the analysis performed in the vacuum case
\cite{AlIoKl2010}, the term $\maxpotr$ is automatically a non-zero
constant. 

As will be indicated in \eqref{eq:decayesternstpot} we have an asymptotic
expansion of $|\ernstpot| \approx M/r$, hence there is some large
radius $R^*$ (which we fix once and for all) such that the following
are true:
\begin{enumerate}
\item $\ernstpot$ does not vanish on $\cauchy \setminus \Phi\circ
B(R^*)$;
\item for every $R > R^*$, on the boundary $\Phi\circ \partial
B(R)$, we have that $|\ernstpot| \approx M/R \geq R^{-2}$.
\end{enumerate}
For $R > R^*$, define
\begin{equation}\label{eq:bootstrapass}
r_0(R) := \inf \left\{ r \in [0,R] ~:~ |\ernstpot| \geq R^{-2} \text{
on } \cauchy \cap \Phi\left[B(R) \setminus B(r)\right] \right\}~.
\end{equation}
Note that by construction $r_0(R) < R^*$ for all $R > R^*$. It
suffices to show that there exists $\tilde{R} > R^*$ such that
$r_0(\tilde{R}) =
0$. We do so by bootstrap: for $\tilde{R}> \sqrt{2} R^*$ sufficiently large, we show that
on $\cauchy \cap \Phi\left[ B(\tilde{R^*}) \setminus B(r_0(\tilde{R}))\right]$ we have in fact the
\emph{improved} estimate 
\[ |\ernstpot| \geq 2 \tilde{R}^{-2}~.\]

\subsection{Asymptotic identities}
To show that the bootstrap assumptions are satisfied near infinity, we
observe that by our assumptions \ref{ass:topology} (which ensures that
$t = \partial_0$ in $\spacetime^{\mathrm{end}}$) and
\ref{ass:asympflat} we can compute the following asymptotic
expansions. (We remark again that below, the parentheses in the
indices denote coordinate evaluation in the coordinates induced by
$\Phi$ introduced in
assumption \ref{ass:topology}.) The inverse metric is given by 
\begin{align*}
g^{(0)(0)} &= -1 - \dfrac{2M}{r} + O_4(r^{-2})~, \\
g^{(0)(i)} &= - 2\sum_{j,k = 1}^3 \lcsym_{ijk} \dfrac{S^jx^k}{r^3} + O_4(r^{-3})~, \\
g^{(i)(j)} &= \delta^{ij} - \dfrac{2M}{r} \delta^{ij} + O_4(r^{-2})~. 
\end{align*}
The Faraday tensor has 
\begin{align*}
H^{(0)(i)} &= \dfrac{q x^i}{r^3} + O_4(r^{-3})~, \\ 
H^{(i)(j)}
& =\dfrac{q}{Mr^3} \sum_{k = 1}^3 \lcsym_{ijk} \left(
\dfrac{3\sum_{l=1}^3S^lx^l}{r^2} x^k - S^k\right) + O_4(r^{-4}) ~,
\end{align*}
which implies that the real part of the potential $\maxpot$ is $O_3(1/r)$ 
and the imaginary part is $O_3(1/r^2)$ (after
normalising to vanish at spatial infinity). This means that
asymptotically $\ernstform$ is given just by the contribution of
$\preernst$, that is
\begin{align*} 
\ernstform_{(0)(j)} &= \dfrac{M}{r^3} x^j + O_3(r^{-3}) +
i \left( \dfrac{1}{r^3} S^j - \dfrac{3\sum_{k=1}^3S^kx^k}{r^5} x^j +
O_3(r^{-4}) \right)~,\\
\ernstform_{(i)(j)} &= \displaystyle \dfrac{1}{r^3} \sum_{k=1}^3 \lcsym_{ijk}S^k -
\dfrac{3\sum_{k=1}^3S^kx^k}{r^5} \sum_{m = 1}^3
\lcsym_{ijm}x^m + O_3(r^{-4}) \\ 
& \qquad + i
\sum_{k = 1}^3 \lcsym_{ijk}\left( \dfrac{M}{r^3} x^k + O_3(r^{-3})\right)~.
\end{align*}
Now we can compute $\ernstpot$: integrating the expression for
$\ernstform_{0j}$ we have that 
\begin{equation}\label{eq:decayesternstpot}
\sigma = - \frac{M}{r} + O_4(r^{-2}) + i \left(
\frac{\sum_{k=1}^3 S^kx^k}{r^3} + O_4(r^{-3}) \right)~.
\end{equation}
This means that $y + iz = - \ernstpot^{-1} = - \bar\ernstpot /
|\ernstpot|^2$ has
\begin{subequations}
\begin{align}
y &= \frac{r}{M} + O_4(1)~,\\
z&= \frac{\sum_{k=1}^3 S^kx^k}{M^2 r} + O_4(r^{-1})~.
\end{align}
\end{subequations}
From above, we compute $\mathfrak{A} = M^2(y^2 + z^2)(\nabla z)^2 + z^2$
(see Lemma \ref{lem:mainlemma} and assumption \ref{ass:kerrnewman}). 
\begin{equation}
\mathfrak{A} = \frac{|S|^2}{M^4} + O_3(r^{-1})~,
\end{equation}
and we remark that $M^2\mathfrak{A}$ converges to $\mathfrak{a}^2$,
the square of total angular momentum (see assumption
\ref{ass:asympflat}). 

We also need to compute $\ernstform^2$. The leading order contribution
comes from \[\sum_{j = 1}^3 (\Re \ernstform_{(0)(j)})^2
g^{(0)(0)}g^{(j)(j)} - \sum_{i,j = 1}^3 (\Im \ernstform_{(i)(j)})^2
g^{(i)(i)}g^{(j)(j)} \approx - \frac{4M^2}{r^4}~.\]
This implies that 
\[
\frac{\ernstform^2}{4\sigma^4} = - \frac{1}{M^2} + O_3(r^{-1})~.
\]
or (see Corollary \ref{cor:identities} and assumption
\ref{ass:kerrnewman})
\begin{equation}\label{eq:errunoasymp}
\errsr{uno} = O_3(r^{-1})~.
\end{equation}

\subsection{Controlling algebraic errors}
Given the behaviour of various quantities at spatial infinity by the
\ref{ass:asympflat} assumption, we can control the quantities in the
interior region by integrating their derivatives from the asymptotic
region. More precisely, we have the following
lemma for \emph{scalar} functions:
\begin{lem}\label{lem:techint}
Let $R_0, \alpha$ be fixed positive reals, and suppose 
that $0 < \delta < R_0^{-(\alpha + 1)}$. Let $f$ be a
function defined on $\Real^3$ such that 
\[ \sum_{j = 1}^3 |\partial_j f| \leq \delta \]
everywhere and 
\[ |f| \leq C r^{-\alpha} \]
on $\Real^3 \setminus B(R_0)$. Then for the same $C$ as above, 
\[ |f| \leq (C+\pi/2) \min (\delta^{\frac{\alpha}{\alpha + 1}}, r^{-\alpha})~.
\]
\end{lem}
\begin{proof}
Since $R_0\delta^{\frac{1}{1+\alpha}} < 1$ by assumption, there exists
$\bar{R} > R_0$ such that $\bar{R} \delta^{\frac{1}{1+\alpha}} = 1$.
To the exterior of $B(\bar{R})$ we have that $|f| \leq C r^{-\alpha}$.
To the interior we have by the fundamental theorem of calculus  
\[ |f(x)| \leq \left| f\left( \frac{x \bar{R}}{|x|} \right)\right| +
\frac{\pi}{2} (\bar{R} -
|x|) \cdot |\partial f| \leq C \bar{R}^{-\alpha} + \frac\pi2 \bar{R} \delta =
(C+\pi/2) \delta^{\frac{\alpha}{\alpha+1}}~.\]
The factor of $\pi/2$ is due to the fact that the straight-line path
between coordinate $x$ and the exterior of $B(\bar{R})$ in the radial
direction may pass through several black-hole regions. Modifying the
paths so that they remain in $\cauchy$ introduces at most a factor of
$\pi/2$ to the path length. 
\end{proof}
\begin{rmk}
The $C + \pi/2$ is not sharp; the sharp estimate depends on optimising
$\pi B / 2
+ C B^{-\alpha}$ for $B$. For the purpose of this paper, it suffices
that $(C+\pi/2)-C$ is a universal constant independent of $\delta$ for
$\delta$ sufficiently small. 
\end{rmk}

Now we are in a situation to prove
\begin{prop}[Main error estimate]\label{prop:errest}
Under the assumptions of the main theorem, there exists a constant
$C_0$ depending only on $M, q, \mathfrak{a}$ and a constant $C_1$
depending on the uniform bound on $g$, $g^{-1}$, the Christoffel
symbols, and $H$ (see assumption \ref{ass:kerrnewman}) such that 
the following estimates are true in $\cauchy \setminus \Phi\circ
B(r_0(R))$ for $R > R^*$:
\begin{subequations}
\begin{align*}
\errsr{uno} &\leq C_0 \min( C_1 \epsilon^{1/2} R^4, r^{-1})\\
\errsr{dos} &\leq C_0 C_1 \epsilon\\
\errsr{tres} &\leq C_0 \min(C_1 \epsilon^{1/2} R^4, r^{-1})\\
\errsr{cuatro} &\leq C_0 \min( C_1\epsilon^{1/2} R^4, r^{-1})\\
\errsr{appone} &\leq C_0 C_1 \epsilon |P_0| \\
\msmaxpot &\leq C_0 \min( C_1 \epsilon^{1/2}, r^{-1}) \\
\left|\mathfrak{A} - \left(\frac{\mathfrak{a}}{M}\right)^2\right|
&\leq C_0 \min( C_1 \epsilon^{1/4} R^6, r^{-1}) 
\end{align*}
\end{subequations}
\end{prop}
\begin{rmk}
The quantities $\errsr{uno}, \errsr{dos}, \errsr{tres}$ are defined in 
Corollary \ref{cor:identities}; the definition and some basic analysis
of $\errsr{cuatro}$ appears in Section \ref{sec:nulldec}; the error
term $\errsr{appone}$ is defined in \eqref{eq:deferrorappone} and
appears in the Main Lemma \ref{lem:mainlemma}; and $\msmaxpot$ is the
potential associated with $\msmax$ as defined in Lemma
\ref{lem:sumpotsisweakerror}. 
\end{rmk}
\begin{proof}
In the following $\lesssim_0, \lesssim_1$ denote that the left hand
side is bounded by the right hand side up to multiplicative constants
$C_0$ and $C_1$ respectively. (The $C_0$, $C_1$ can change from line
to line in the proof.)

For $\errsr{uno}$, by the defining condition \eqref{eq:bootstrapass}
for $r_0(R)$ (upon whose value we will bootstrap), by Lemma
\ref{lem:unoderiserror}, and by the assumption \ref{ass:kerrnewman}, we have 
\[ |\partial \errsr{uno}| \lesssim_1 \epsilon R^8 \]
and the decay condition
\[ |\errsr{uno}| \lesssim_0 r^{-1} \]
which implies by Lemma \ref{lem:techint} 
\[ |\errsr{uno}| \lesssim_0 \min( C_1 \epsilon^{1/2} R^4, r^{-1})~.\]
This immediately implies the same bound for $\errsr{cuatro}$. (See
Section \ref{sec:nulldec}.) 

For $\errsr{dos}$, it follows directly from the definition that 
\[ |\errsr{dos}| \lesssim_1 \frac{\epsilon}{M}~. \]
Similarly, $\errsr{appone}$ can be directly bounded by $\frac{|P_0|}{M^2}C_1 \epsilon$. 

For $\msmaxpot$, its derivative is a direct error term, hence
$|\partial \msmaxpot| \leq C_0 C_1 \epsilon$. Its decay rate is $C_0 /
r$, which implies by Lemma \ref{lem:techint} that
\[ |\msmaxpot| \lesssim_0 \min( C_0 C_1 \epsilon^{1/2}, r^{-1})~. \]

An estimate for $\errsr{tres}$ can be directly obtained from the
estimate for $\msmaxpot$, if we use the bootstrap assumption
\eqref{eq:bootstrapass}. However, this will lead to a term where $R$
is not paired against $\epsilon$, which will cause difficulties for
closing the bootstrap argument. Instead, we estimate it directly from
its derivatives: from the product rule we have that
\[ |\partial \errsr{tres}| \leq C_0 C_1 R^8 \epsilon ~.\]
On the other hand, we know that the asymptotic behaviour of
$\errsr{tres}$ can be read-off from \eqref{eq:errunoasymp} and that of
$\msmaxpot$, that is asymptotically $|\errsr{tres}| \lesssim_0 r^{-1}$.
This implies via our technical lemma again
\[ |\errsr{tres}| \lesssim_0 \min( C_0 C_1 R^4 \epsilon^{1/2}, r^{-1})
~.\]

Lastly we estimate $\mathfrak{A}$. From the asymptotic behaviour
computed in the previous section, we have that at the asymptotic end
$\mathfrak{A} - (\mathfrak{a}/M)^2 \lesssim_0 r^{-1}$. Its derivative
we estimate using Lemma \ref{lem:mainlemma}, where the following
points are observed:
\begin{itemize}
\item The terms $y,z$ are size $\sigma^{-1}$ or $R^2$. 
\item The terms $\nabla y$ and $\nabla z$ are size
$\frac{1}{|\ernstpot|^2}\nabla \bar{\ernstpot}$ or $C_1 R^4$. 
\item The term $\msmaxpot$ we (roughly) estimate by $C_0 C_1
\epsilon^{1/2}$.
\item The term $\errsr{uno}$ we (roughly) estimate by $C_0 C_1
\epsilon^{1/2}R^4$. 
\end{itemize}
This gives us
\begin{align*}
|\nabla_a\mathfrak{A}| &\leq \Bigg|
\frac{4\mu^2}{|\ernstpot|^2}\nabla^bz\Im\left(
\frac{t^c}{\ernstpot^2\bar\maxpotr}\nabla_a\msmax_{cb} -
\frac{\mu}{\ernstpot^2\bar\maxpotr}\msweyl_{dacb}t^ct^d\right)  + 2
\nabla_az
\Im\left(\frac{\bar\kappa\bar\ernstpot}{\mu^2\ernstpot}\msmaxpot\right)
\\
& \qquad + \mu^2 t^2(z\nabla_a y - y\nabla_a z)\Im\errsr{uno} -
\Im\left[\frac{2\errsr{uno}\mu^2}{|\ernstpot|^2}\nabla_az\left(\ernstpot
t^2 + \frac{i}{\mu}\Im(\bar\ernstpot^2\maxpotr)\right)\right] \\
\notag & \qquad - \frac{z\nabla_az}{\mu^2}(|\msmaxpot -
\kappa\ernstpot|^2 -
|\kappa\ernstpot|^2)  +
\Im\left[\frac{4\mu^3}{|\ernstpot|^2\ernstpot^2\bar\maxpotr}\nabla^bz
(\errsr{appone})_{ab}\right] \\
\notag & \qquad +
\Im\left[\frac{4\mu}{|\ernstpot|^2\ernstpot^2}\ernstform_{cb}\Re(\bar\maxpot\msmax_a{}^c)\nabla^bz
-
\frac{\mu\maxpotr\bar\ernstpot}{\ernstpot}\Im(\errsr{uno})\nabla_a\ernstpot^{-1}\right]\Bigg|\\
&\leq C_0 C_1 \left[ R^{12} \epsilon + R^4 \epsilon^{1/2} + 
R^{10} \epsilon^{1/2} + R^{10} \epsilon^{1/2} + R^6 \epsilon^{1/2} +
R^{12} \epsilon + R^{12} \epsilon + R^{8} \epsilon^{1/2} \right]\\
& \leq C_0 C_1 R^{12} \epsilon^{1/2}
\end{align*}
where we used that $\epsilon$ will be small and $R$ large. Integrating
using Lemma \ref{lem:techint} we get
\[ \left|\mathfrak{A} - \left(\frac{\mathfrak{a}}M\right)^2\right| \leq
C_0 \min( C_1 R^{6} \epsilon^{1/4}, r^{-1})~. \]
\end{proof}

Applying the above estimates to Corollary \ref{cor:identities}, we
obtain immediately the following
\begin{cor}\label{cor:almostidentities}
The following almost identities are true:
\begin{subequations}
\begin{align}
\left| \Box y + \frac{2}{M^2} \frac{1 - y}{y^2 + z^2}\right| & \leq
C_0 C_1 R^4\epsilon^{1/2} \\
\left| (\nabla z)^2 - \frac{\frac{\mathfrak{a}^2}{M^2} - z^2}{M^2(y^2
+ z^2)} \right| & \leq C_0 C_1 R^6 \epsilon^{1/4} \\
\left| (\nabla y)^2 - \frac{\frac{\mathfrak{a}^2}{M^2} +
\frac{q^2}{M^2} + y^2 - 2y}{M^2(y^2 + z^2)} \right| &\leq C_0 C_1 R^6
\epsilon^{1/4}
\end{align}
\end{subequations}
\end{cor}

\subsection{Closing the bootstrap}
To close the bootstrap, that is, to obtain the improved decay estimate
$|\ernstpot| \geq 2 \tilde{R}^{-2}$ for sufficiently small $\epsilon$ and
sufficiently large $\tilde{R}$ on the domain $E_{\tilde{R}} :=
\cauchy \cap \Phi\left[ B(R^*)
\setminus B(r_0(\tilde{R}))\right]$, it suffices to consider the domain
$W_{\tilde{R}} := E_{\tilde{R}} \cap \{ |\ernstpot| \leq 4 \tilde{R}^{-2} \}$.  
Consider first \eqref{eq:identity2p}. By studying the asymptotic
limit, we have that the constant term is 0. On $W_{\tilde{R}}$ then we
have 
\[ \left|t^2 + 1\right| \leq C_0 \tilde{R}^{-2} + C_0 C_1 \epsilon^{1/2}~. \]
So for sufficiently large $\tilde{R} > 3R^*$ (now depending on $C_0$) and
sufficiently small $\epsilon$ (now depending on $C_0$ and $C_1$) we
have that $t^2 < -1/2$. In particular the Killing vector field is
time-like. Now using that $t(y) = t(z) = 0$, we have that $\nabla y$
and $\nabla z$ are \emph{space-like} in $W_{\tilde{R}}$. 

Since $E_{\tilde{R}}$ has compact closure, we have that
$W_{\tilde{R}}$ has compact closure. Using that $t^2 \leq -1/2$ on
this set, we have that $\sum_{i = 1}^3|\partial_{i}\ernstpot^{-1}| \leq
C_1\left[ |(\nabla z)^2| + |(\nabla y)^2|\right]$. The right hand side
we bound by Corollary \ref{cor:almostidentities}, and the fact that in
$W_{\tilde{R}}$ we have the upper bound $(y^2 + z^2)^{-1} =
|\ernstpot|^2 \leq 16 \tilde{R}^{-4}$. This leads to 
\begin{equation}
\sum_{i = 1}^3 |\partial_{i}\ernstpot^{-1}| \leq C_0 C_1 ( 1 +
\tilde{R}^{-4} +
\tilde{R}^{6}\epsilon^{1/4} )
\end{equation}
so by the fundamental theorem of calculus, integrating from the boundary
of $W_{\tilde{R}}$ 
where $|\ernstpot| \geq 4 \tilde{R}^{-2}$, 
\begin{align*}
|\ernstpot^{-1}| &\leq \frac14 \tilde{R}^2 + C_0 C_1(1 + \tilde{R}^{-4}
+ \tilde{R}^6\epsilon^{1/4}) R^* \\
& \leq \frac14 \tilde{R}^2 + C_0 C_1
\tilde{R} + C_0 C_1 \tilde{R}^{-3} + C_0 C_1 \tilde{R}^{7}
\epsilon^{1/4} 
\end{align*}
where the $R^*$ denotes the maximum coordinate distance one has to
integrate (since $W_{\tilde{R}} \subseteq \Phi\circ B(R^*)$). By
choosing $\tilde{R}$ sufficiently large, and 
\begin{equation}\label{eq:relsizeepsilon}
\epsilon^{1/4} \ll \tilde{R}^{-6}~,\end{equation}
we can bound the right hand side
\begin{equation}
|\ernstpot^{-1}| \leq \frac12 \tilde{R}^2
\end{equation}
as desired. 

\begin{rmk}
The value $\tilde{R} > R^* > R_0$ is chosen to be sufficiently large
relative to the constants $C_0$ and $C_1$ measuring the sizes of the
asymptotic $M, q, \mathfrak{a}$ and uniform bounds on the metric etc.
The value $\epsilon$ is now required to be sufficiently small relative
to $C_0$, $C_1$, and $\tilde{R}$, which implies that $\epsilon$ only
needs to be sufficiently small relative to $C_0$ and $C_1$. See also
assumption \ref{ass:kerrnewman}. 
\end{rmk}
\begin{rmk}\label{rmk:fixingr}
After the bootstrap argument above, $\tilde{R}$ should be considered
a fixed constant depending on $C_0$ and $C_1$. That is to say, it is
understood that the right hand sides of the almost identities in 
Corollary \ref{cor:almostidentities} can be made arbitrarily small by
choosing sufficiently small $\epsilon$. 
\end{rmk}

\section{Proof of the Main Theorem}
Now that we know the function $y$ can be smoothly defined on the
entirety of our partial Cauchy surface $\cauchy$ and extended smoothly 
past the horizons $\horizon^0$, 
we can study the local behaviour of $y$ near a bifurcate sphere
$\horizon^0_i$. We
will, in fact, demonstrate that 
\begin{itemize}
\item $y$ is almost constant on the bifurcate sphere, and
\item $y$ increases as we move off the horizon. 
\end{itemize}
One expects that, given that the local deviation of our space-time
from the Kerr-Newman solutions is not too large (as required by
assumption \ref{ass:kerrnewman}; see also Theorem
\ref{thm:character}), the constant which approximates $y$ on the
bifurcate sphere is $\frac{1}{M}\left( M + \sqrt{M^2 -
\mathfrak{a}^2 - q^2}\right)$, the value taken by $y$ on the
corresponding Kerr-Newman black hole. For the Kerr-Newman solution, this
value is also the largest value of $y$ at which the function $y$ can
attain a critical point; this is captured in Lemma
\ref{lem:exactmars}. In the case under consideration in this paper, we
instead use the approximate identities of Corollary
\ref{cor:almostidentities} to conclude that at critical points of the
function $y$, the value of $y$ cannot be too much greater than its
value on the horizon. Together with the above two bullet points and a
mountain-pass lemma, we can derive a conclusion which morally states 
that $y$ cannot have a
critical point in the domain of outer communications, and hence there
must only be one black hole. 

In the sequel we implement the above heuristics in detail.

\subsection{Near horizon geometry}
We wish to study the behaviour of $y$ near the bifurcate spheres;
without loss of generality we consider a small neighborhood of
$\horizon_1^0$ in $\spacetime$ (see Assumption
\ref{ass:smoothbifsphere} for definitions). We begin by establishing
a double null foliation of the neighborhood and briefly recalling some
implications of a non-expanding horizon (for more detailed discussion
please see \cite{AlIoKl2010, AlIoKl2010a, Wong2009a}). In the sequel we will
always implicitly work in a small neighborhood of $\horizon_1^0$,
whose smallness depends on $M,q,\mathfrak{a}$, and the uniform bounds
on the metric, its inverse, the Christoffel symbols, and the Faraday
tensor in Assumption \ref{ass:kerrnewman}, but independent of the
smallness parameter $\epsilon$.

Along $\horizon_1^\pm$ let $L^\pm$ be future-directed geodesic generators of
the respective null hypersurfaces. We choose to normalise 
$g(L^+,L^-) = -1$ on $\horizon_1^0$. Along $\horizon_1^\pm$ we define
the functions $u^\mp$ by $L^\pm(u^\mp) = 1$ and $u^\mp
|_{\horizon_1^0} = 0$. The level sets of $u^\mp$ are topological
spheres, and are space-like surfaces. Extend $L^\mp$ to
$\horizon_1^\pm$ to be the unique future-directed null vector orthogonal to
the level sets of $u^\mp$ and satisfying $g(L^-, L^+) = -1$. Now extend
$L^\mp$ off $\horizon_1^\pm$ geodesically, and declare $L^\pm(u^\pm) =
0$. This defines a double-null foliation $u^\pm$ with associated null
vector fields $L^\pm$ in the neighborhood of $\horizon_1^0$. 

Along $\horizon_1^\pm$ the null second fundamental
form $g(\covD_X L^\pm, Y) = - g(L^\pm,\covD_X Y)$ 
(for $X,Y$ vector fields tangent to $\horizon_1^\pm$) vanishes
identically due to the horizons being non-expanding (see,
e.g.~\cite[\S 2.5]{Wong2009a}). This implies that $\preernst\cdot L^\pm \propto
L^\pm$ along the horizons:
\[ \Re \preernst(X,L^\pm) = g(\nabla_X t, L^\pm) = 0~,\]
and the imaginary part follows once it is realised that the Hodge dual
of $L^\pm\wedge X$ can be written as $L^\pm \wedge Y$ for some $Y$
also tangent to $\horizon_1^\pm$. Furthermore, Raychaudhuri's equation
then guarantees that $\maxw\cdot L^\pm \propto L^\pm$ along the
horizon, using that $\Ricci(L^\pm, L^\pm) = (\maxw\cdot L^\pm)_a
(\bar\maxw\cdot L^\pm)^a$ \cite[\S 2.5]{Wong2009a}.
Together these imply (via the definition \eqref{eq:def:ernstform}) 
that $L^\pm$ are in fact the null principal directions of 
$\ernstform$ on $\horizon_1^0$.  

Furthermore, observe that since $t^a$ is tangent to $\horizon_1^\pm$
which intersect transversely, we must have $t^a$ is tangent to
$\horizon_1^0$. This implies that $g(L^\pm,t) = 0$ along
$\horizon_1^0$. 

\begin{prop}\label{prop:yhorizonvalue}
For $\epsilon$ sufficiently small, along $\horizon_1^0$, 
\[ \left| M y - \left( M + \sqrt{M^2 - \mathfrak{a}^2 - q^2}\right) \right| \lesssim \epsilon^{1/4}~.\]
\end{prop}
\begin{rmk}\label{rmk:rootquadratic}
The quadratic polynomial $y^2 - 2y + \frac{\mathfrak{a}^2 + q^2}{M^2}$
plays a recurring role in our argument. We note that the two roots to
the polynomial are
\[ y_{\pm} = \frac{1}{M} \left( M\pm \sqrt{M^2 - \mathfrak{a}^2 -
q^2}\right)~.\]
That we need to ensure the existence of two distinct roots, one larger
than, and one smaller than $1$ is why sub-extremality is assumed in
\ref{ass:asympflat}. (Of course, the extremal Kerr-Newman black holes
have very different horizon geometry, and we should not expect an
analysis based on the bifurcate spheres to carry over in that case.)
\end{rmk}
\begin{rmk}
The proposition and its proof are largely the same as Lemma 4.1 in
\cite{AlIoKl2010}; we sketch the proof here for completeness.
\end{rmk}
\begin{proof}
Since $L^\pm$ along $\horizon_1^0$ are the null principal directions
of $\ernstform$, we can apply the results of Section
\ref{sec:nulldec}. In particular, we have that the orthogonality of
$L^\pm$ to the Killing vector field $t$ on the horizons
implies the \emph{exact identity} (that the following two equations do
not contain error terms is very important in the sequel)
\begin{subequations}
\begin{align}
L^+(y) = L^+(z) &= 0 \text{ on } \horizon_1^+~,\\
L^-(y) = L^-(z) &= 0 \text{ on } \horizon_1^-~.
\end{align}
\end{subequations}
These imply that on $\horizon_1^0$
\begin{equation}\label{eq:erryhorizon}
\nabla_a y = \Re \left[ i\errsr{cuatro}
\lcsym_{bacd}t^b(L^-)^c(L^+)^d\right]
\end{equation}
is of size $\epsilon^{1/2}$ by Proposition \ref{prop:errest} and
Remark \ref{rmk:fixingr}. This implies that $(\nabla y)^2 =
O(\epsilon^{1/4})$. So using Corollary \ref{cor:almostidentities} we
obtain that along the horizon 
\[ \frac{ \frac{\mathfrak{a}^2 + q^2}{M^2} + y^2 - 2y}{M^2(y^2 + z^2)}
= O(\epsilon^{1/4})~. \]
By the bootstrap argument, we have that $(y^2 + z^2)^{-1}$ is bounded
above by a constant depending only on $C_0, C_1$ (see Remark
\ref{rmk:fixingr} again), hence we have that on $\horizon_1^0$ 
\[ y^2 - 2y + \frac{\mathfrak{a}^2 + q^2}{M^2} = O(\epsilon^{1/4})~.\]

Observe further that by \eqref{eq:erryhorizon}, if $X,Y$ are vector fields
tangent to $\horizon_1^0$, we have that 
\[ X(Y(y)) = \Re \left[ i X(\errsr{cuatro}) \lcsym(t,Y,L^-,L^+) + i
\errsr{cuatro} X(\lcsym(t,Y,L^-,L^+))\right]~.\]
From the definition of $\errsr{cuatro}$ in Section \ref{sec:nulldec},
we see immediately that $\covD_a \errsr{cuatro}$ can be controlled by
$\errsr{uno}$ and $\covD_a \errsr{uno}$. That is to say, we have that
the Hessian of $y$ \emph{along} $\horizon_1^0$ is also of order
$\epsilon^{1/4}$. 

This gives two possibilities: either $|y - y_+| \lesssim
\epsilon^{1/4}$ or $|y - y_-|\lesssim \epsilon^{1/4}$; it suffices to
eliminate the second alternative. To do so we consider the first
inequality in Corollary \ref{cor:almostidentities}. Provided
$\epsilon$ is sufficiently small (especially compared to $\sqrt{M^2 -
\mathfrak{a}^2 - q^2}$), that $|y - y_-| \lesssim \epsilon^{1/4}$
along $\horizon_1^0$ would imply $\Box y < 0$ in a small neighborhood
of the bifurcate sphere. We use this fact to show that $y$ must
decrease as we move off the horizon.  

Define $\tilde{y}$ by setting $\tilde{y} = y$ along $\horizon_1^-$,
and requiring that $L^+\tilde{y} = 0$. This guarantees that in a small
neighborhood of $\horizon_1^0$, $\tilde{y}$ is bounded by
$\sup_{\horizon_1^0} y$. Using that the Hessian of $y$ tangent to
$\horizon_1^0$ is also an error term, this implies that $|\Box
\tilde{y}| \lesssim \epsilon^{1/4}$; that is to say, the main
contribution to $\Box y$ comes from $L^-(L^+ y)$. Using that $y$ and
$\tilde{y}$ agree on $\horizon_1^{\pm}$, we can write $y = \tilde{y} +
u^+ u^- \hat{y}$ where $\hat{y}$ is a smooth function in a small
neighborhood of $\horizon_1^0$. Furthermore, on $\horizon_1^0$ we have
that $\Box (y - \tilde{y}) = - 2 \hat{y}$, hence along $\horizon_1^0$ we
have 
\[  \left|\hat{y} - \frac{1-y}{M^2(y_-^2 + z^2)} \right| \lesssim
\epsilon^{1/2} \]
and in particular for all $\epsilon$ sufficiently small 
\[ \hat{y}|_{\horizon_1^0} \geq \frac{1- y_-}{2M^2(y^2 + z^2)} > 2C_h >
0 ~.\]
By continuity, on a sufficiently small neighborhood of $\horizon_1^0$
we have that $\hat{y} \geq C_h$. Now using that in the domain of
outer communications, by construction we have $u^+u^- < 0$, this
implies that
\[ y \leq \tilde{y} + u^+u^- \hat{y} \leq y_- + O(\epsilon^{1/4}) -
|u^+u^-| C_h \]
in the small neighborhood of $\horizon_1^0$. Now consider all points
in this neighborhood for which $-u^+u^- \geq \delta > 0$ for some
fixed $\delta$. Then for all $\epsilon$ sufficiently small, 
at these points we have $y < y_- - \frac12 C_h\delta$. By the asymptotic behaviour of $y$ (growing
to $+\infty$), this implies that $y |_{\cauchy\cap\outercom}$ achieves
a minimum value that is at most $y_- - \frac12 C_h \delta$. But this implies
(using that $t^a$ is transverse to $\cauchy\cap\outercom$) that $y$
attains a critical point at a value $y_- - \frac12 C_h\delta$, which is
impossible for sufficiently small $\epsilon$ by Corollary
\ref{cor:almostidentities}. This concludes the proof that $y$ must be
close to $y_+$ on the horizon. 
\end{proof}

\begin{rmk}\label{rmk:increasingoffhorizon}
The same argument in the contradiction step of the proof can be used
to show that, given $y$ is close to $y_+$ on the horizon,  there 
exists some topological sphere in $\cauchy\cap\outercom$ that 
encloses $\horizon_1^0$ and some $\delta > 0$ ($\delta$ depends on
$M,q,\mathfrak{a}$, and the uniform bounds on the metric, its inverse,
its Christoffel symbols, and the Faraday tensor) such that restrict to
that sphere $y > y_+ + 2\delta > \sup_{\horizon_1^0} y + \delta$
provided $\epsilon$ is sufficiently small. 

In particular, we define $\hat{y}$ as above. But now using that $y
\approx y_+$ on the horizon we have that for all $\epsilon$
sufficiently small 
\[ \hat{y} |_{\horizon_1^0} \leq \frac{ 1 - y_+}{2M^2(y_+^2 + z^2)} <
-2C_h < 0 \]
which allows us to conclude that
\[ y \geq \tilde{y} + u^+ u^- \hat{y} \geq y_+ - O(\epsilon^{1/4}) +
|u^+u^-| C_h ~.\]
Choosing $2\delta$ sufficiently small to be attained by $|u^+u^-|C_h$,
then choosing $\epsilon$ even smaller we get that $y$ would increase
to at least $y_+ + \delta$ off the horizon. 
\end{rmk}

\subsection{Concluding the proof}
Having established our technical results about the behaviour of $y$
near the horizon sphere $\horizon_1^0$ (and hence by symmetry for any
$\horizon_i^0$), we conclude our main theorem by appealing to a finite
dimensional mountain pass lemma (see Appendix \ref{appendix2}). 

\begin{proof}[Proof of Theorem \ref{thm:mainthm}]
Assume, for contradiction, that there are at least two black holes. By
Proposition \ref{prop:yhorizonvalue} and Remark
\ref{rmk:increasingoffhorizon} we know that for sufficiently small
$\epsilon$, we can find $\delta > 0$ such that $y|_{\horizon^0} < y_+
+ \delta$ and there exists a topological sphere $S\subset\cauchy\cap\outercom$ 
(using that we have a lower bound on the coordinate-distance between
$\horizon_1^0$ and $\horizon_2^0$; see \ref{ass:topology}) such that $\horizon_1^0$ and $\horizon_2^0$ are in disjoint subsets of
$\cauchy \setminus S$ and such that $y|_S > y_+ + 2 \delta$. 
By the
asymptotic growth of $y$ we know that $y$ satisfies the Palais-Smale
condition. So applying Lemma \ref{lem:mpt} to the function $y$ on the
manifold $(\cauchy\cap\outercom)\cup \horizon^0$, $y$ attains a 
critical point in
$\cauchy\cap\outercom$ where the value of $y$ is at least $y_+ +
2\delta$. Using that $t^a$ is transverse to $\cauchy\cap\outercom$,
again we have that $\nabla y = 0$ there. For sufficiently small
$\epsilon$ this leads to a contradiction with Corollary
\ref{cor:almostidentities} together with Remark
\ref{rmk:rootquadratic}.  
\end{proof} 

\appendix

\numberwithin{equation}{section}

\section{Proof of the Main Lemma}\label{appendix1}
In this appendix we shall give the proof of Lemma \ref{lem:mainlemma},
which claims that $\mathfrak{A} = \mu^2(y^2+z^2)(\nabla z)^2 + z^2$ is
``almost constant''. We start directly with the definition
\begin{equation}\label{eq:nablaA} \nabla_a \mathfrak{A} = 2\mu^2(y\nabla_ay + z\nabla_az)(\nabla z)^2
+ 2z\nabla_a z + 2\mu^2(y^2 + z^2)\nabla^bz \nabla_a\nabla_b
z~.\end{equation}
The focus will be on the third term in the expansion, which contains
the Hessian of $z$. Therefore we compute $\nabla^2_{a,b}\ernstpot^{-1}$.
\[
\nabla_a\nabla_b \ernstpot^{-1} = -\nabla_a
(\ernstpot^{-2}\nabla_b\ernstpot) = 2\ernstpot\nabla_a\ernstpot^{-1}\nabla_b\ernstpot^{-1} -
\ernstpot^{-2}\nabla_a\nabla_b\ernstpot
\]
Next use
\begin{align*}
\nabla_a\nabla_b\ernstpot &= \nabla_a \ernstform_{cb}t^c \\
& = \ernstform_{cb}\nabla_at^c + \frac{t^c}{\bar\maxpotr}\left(
2\nabla_a(\bar\maxpot\msmax_{cb}) -
2\kappa\nabla_a\bar\maxpot\ernstform_{cb}\right)\\
& \qquad -\frac{2\mu t^ct^d}{\bar\maxpotr}\left(\weyls_{dacb}
+ \left(\Ricci_d{}^eg_a{}^f -
\Ricci_a{}^e g_d{}^f\right)\Iasd_{efcb}\right)
\end{align*}
We can expand $\Iasd$ by the definition, use Einstein's equation
\eqref{eq:einsteq} to replace the Ricci tensor, and use the
definitions \eqref{eq:def:msmax} and
\eqref{eq:def:msweyl} to obtain that
\begin{align*}
\nabla_a\nabla_b \ernstpot &-
\frac{2t^c}{\bar\maxpotr}\bar\maxpot\nabla_a\msmax_{cb} +
\frac{2\mu}{\bar\maxpotr}\msweyl_{dacb}t^ct^d \\
&= \frac12 \ernstform_{cb}\preernst_a{}^c +
\frac12\ernstform_{cb}\bar\preernst_a{}^c +
\frac{4\mu}{\bar\maxpotr}\nabla_a\bar\maxpot\nabla_b\maxpot 
+ \frac{3}{\ernstpot}(\ernstform\tilde{\otimes}\ernstform)_{dacb}t^dt^c
\\
& \quad - \frac{2\mu t^ct^d}{\bar\maxpotr}\left( \maxw_{dl}\bar\maxw_c{}^l
  g_{ab} - \maxw_{dl}\bar\maxw_b{}^lg_{ac} -
\maxw_{al}\bar\maxw_c{}^lg_{db} +
\maxw_{al}\bar\maxw_b{}^lg_{cd}\right)\\
&\quad - \frac{2\mu t^ct^d}{\bar\maxpotr}\left(
i\bar\maxw^{el}\lcsym_{eacb}\maxw_{dl}-
i\maxw^{el}\lcsym_{edcb}\maxw_{al}\right)~.
\end{align*}
For the terms in the last line, we can use the identity for self-dual
two-forms
\begin{equation}\label{eq:sdhdidentity}
i\bar{\mathcal{X}}^{kh}\lcsym_{wyzk} = g^h_w\bar{\mathcal{X}}_{yz} +
g^h_y\bar{\mathcal{X}}_{zw} + g^h_z\bar{\mathcal{X}}_{wy} 
\end{equation}
which gives
\begin{align*}
- \frac{2\mu t^ct^d}{\bar\maxpotr}&\left(
i\bar\maxw^{el}\lcsym_{eacb}\maxw_{dl}-
i\maxw^{el}\lcsym_{edcb}\maxw_{al}\right)\\
&= - \frac{2\mu t^ct^d}{\bar\maxpotr}\left( \maxw_{ac}\bar\maxw_{bd} -2\maxw_{da}\bar\maxw_{cb}
- \maxw_{db}\bar\maxw_{ac} - \maxw_{dc}\bar\maxw_{ba} +
  \maxw_{ab}\bar\maxw_{dc} \right)
\end{align*}
where by (anti)symmetry, after the contraction against $t^ct^d$, the
last two terms in the parenthesis evaluate to zero. Hence we can
simplify
\begin{align*}
\nabla_a\nabla_b \ernstpot &-
\frac{2t^c}{\bar\maxpotr}\bar\maxpot\nabla_a\msmax_{cb} +
\frac{2\mu}{\bar\maxpotr}\msweyl_{dacb}t^ct^d \\
&= \frac12 \ernstform_{cb}\preernst_a{}^c +
\frac12\ernstform_{cb}\bar\preernst_a{}^c +
\frac{4\mu}{\bar\maxpotr}\nabla_a\bar\maxpot\nabla_b\maxpot 
+
\frac{3}{\ernstpot}(\ernstform\tilde{\otimes}\ernstform)_{dacb}t^dt^c
\\
& \quad - \frac{2\mu}{\bar\maxpotr}\left(
\nabla\maxpot\cdot\nabla\bar\maxpot g_{ab} +
\maxw_{al}\bar\maxw_b{}^lt^2
- \nabla^l\maxpot\bar\maxw_{bl}t_a -
  \nabla^l\bar\maxpot\maxw_{al}t_b\right)\\
&\quad - \frac{2\mu}{\bar\maxpotr}\left(
\nabla_b\maxpot\nabla_a\bar\maxpot -
\nabla_a\maxpot\nabla_b\bar\maxpot\right)~.
\end{align*}
In the following we will also group terms proportional to $t_b$ on the
left-hand-side of the expression, since in \eqref{eq:nablaA}, the
$\nabla_a\nabla_bz$ term is multiplied against $\nabla^bz$, and we
have that $t_b\nabla^bz = 0$ by our assumption that $t$ is a symmetry. 

Directly expanding the terms 
\begin{align*}
(\ernstform\tilde{\otimes}\ernstform)_{dacb}t^ct^d &=
\ernstform_{da}\ernstform_{cb}t^ct^d - \frac13
\Iasd_{dacb}\ernstform^2 t^dt^c \\
&= \ernstpot^4\nabla_a\ernstpot^{-1}\nabla_b\ernstpot^{-1} -
\frac1{12}\ernstform^2t^2g_{ab} +
\frac1{12} \ernstform^2t_at_b~,
\end{align*}
we arrive at
\begin{align}
\notag\nabla_a\nabla_b\ernstpot^{-1} &+
\frac{2t^c}{\ernstpot^2\bar\maxpotr}\nabla_a\msmax_{cb} -
\frac{2\mu}{\ernstpot^2\bar\maxpotr}\msweyl_{dacb}t^ct^d +
\frac{1}{4\ernstpot^3}\ernstform^2t_at_b +
\frac{2\mu}{\ernstpot^2\bar\maxpotr}\nabla^l\bar\maxpot\maxw_{al}t_b\\
\label{eq:appone}& = - \ernstpot\nabla_a\ernstpot^{-1}\nabla_b\ernstpot^{-1} +
\frac{1}{4\ernstpot^3}\ernstform^2 t^2 g_{ab} -
\frac{1}{2\ernstpot^2}\ernstform_{cb}\left(\preernst_a{}^c +
\bar\preernst_a{}^c\right)\\
\notag & \quad -
\frac{2\mu}{\ernstpot^2\bar\maxpotr}\left(\nabla_a\maxpot\nabla_b\bar\maxpot
+ \nabla_b\maxpot \nabla_a\bar\maxpot -
\nabla\maxpot\cdot\nabla\bar\maxpot g_{ab} +
\nabla^l\maxpot\bar\maxw_{bl}t_a - \maxw_{al}\bar\maxw_b{}^l
t^2\right)~.
\end{align}

To apply to \eqref{eq:nablaA}, we next multiply \eqref{eq:appone} by
$\nabla^b z = -\Im \nabla^b\ernstpot^{-1}$. We first consider the terms
on the last line, where the expression inside the parenthesis is
real-valued. So we can consider multiplication by
$\nabla\ernstpot^{-1}$ instead of by $\nabla z$. Observe that
\begin{align*}
\nabla_a\maxpot\nabla_b\bar\maxpot & + \nabla_b\maxpot \nabla_a\bar\maxpot -
\nabla\maxpot\cdot\nabla\bar\maxpot g_{ab} + \nabla^l\maxpot\bar\maxw_{bl}t_a - 
\maxw_{al}\bar\maxw_b{}^l t^2\\
&= \maxw^{pr}\bar\maxw^{qs}t^mt^n\cdot (g_{ap}g_{bq}g_{rm}g_{sn} +
g_{bp}g_{aq}g_{rm}g_{sn} \\
& \qquad \qquad - g_{pq}g_{ab}g_{rm}g_{sn} - g_{ap}g_{bq}g_{mn}g_{rs} - g_{bq}g_{rs}g_{an}g_{pm}) 
\intertext{since the last two terms in the parenthesis has a $g_{rs}$
product, we can apply \eqref{eq:asdprodprop1} to swap the $p$ and $q$
indices}
&= \maxw^{pr}\bar\maxw^{qs}t^mt^n\cdot (g_{ap}g_{bq}g_{rm}g_{sn} +
g_{bp}g_{aq}g_{rm}g_{sn} \\
& \qquad \qquad - g_{pq}g_{ab}g_{rm}g_{sn} - g_{aq}g_{bp}g_{mn}g_{rs}
- g_{bp}g_{rs}g_{an}g_{qm}) \\
&= \left(\frac{\kappa\bar\kappa}{4\mu^2}
\ernstform^{pr}\bar\ernstform^{qs} - \frac{\kappa\bar\kappa}{4\mu^2}
\ernstform^{pr}\bar\ernstform^{qs} + \maxw^{pr}\bar\maxw^{qs} \right)
t^mt^n \cdot (g_{ap}g_{bq}g_{rm}g_{sn} \\
& \qquad \qquad +
g_{bp}g_{aq}g_{rm}g_{sn} - g_{pq}g_{ab}g_{rm}g_{sn} -
g_{aq}g_{bp}g_{mn}g_{rs} - g_{bp}g_{rs}g_{an}g_{qm})
\end{align*}
Inside the first parenthesis, we have that $-
\frac{\kappa\bar\kappa}{4\mu^2}
\ernstform^{pr}\bar\ernstform^{qs} + \maxw^{pr}\bar\maxw^{qs}$ is an
error term by using \eqref{eq:def:msmax}. So we
define the \emph{algebraic error term}
\begin{multline}\label{eq:deferrorappone}
(\errsl{appone})_{ab} = \left(\maxw^{pr}\bar\maxw^{qs}-
\frac{\kappa\bar\kappa}{4\mu^2}
\ernstform^{pr}\bar\ernstform^{qs}\right)t^mt^n \cdot
(g_{ap}g_{bq}g_{rm}g_{sn} \\
+ g_{bp}g_{aq}g_{rm}g_{sn} - g_{pq}g_{ab}g_{rm}g_{sn} -
g_{aq}g_{bp}g_{mn}g_{rs} - g_{bp}g_{rs}g_{an}g_{qm})
\end{multline}

We next consider the left-over term given by
$\ernstform^{pr}\bar\ernstform^{qs}$. Using that $\nabla_b\ernstpot^{-1} =
\ernstpot^{-2}\ernstform_{ub}t^u$, we consider
\begin{multline*}
\ernstform^{pr}\bar\ernstform^{qs}\ernstform^{bu}t_ut^mt^n  (g_{ap}g_{bq}g_{rm}g_{sn}  
\\
+ g_{bp}g_{aq}g_{rm}g_{sn} -
g_{pq}g_{ab}g_{rm}g_{sn} -
g_{aq}g_{bp}g_{mn}g_{rs} - g_{bp}g_{rs}g_{an}g_{qm})
\end{multline*}
The first and the third terms inside the parenthesis cancel each
other. We can use product property \eqref{eq:asdprodprop2} with
$g_{bp}$ to obtain
\[
\frac14\ernstform^2 t^rt^mt^n \bar\ernstform^{qs} (g_{aq}g_{rm}g_{sn}
- g_{aq}g_{mn}g_{rs} - g_{rs}g_{an}g_{qm})~. \]
The first two terms cancel each other, and the third vanishes as
$\bar\ernstform$ is antisymmetric. From this we conclude that
\[
\nabla^bz\left(\nabla_a\maxpot\nabla_b\bar\maxpot  + \nabla_b\maxpot
\nabla_a\bar\maxpot -
\nabla\maxpot\cdot\nabla\bar\maxpot g_{ab} +
\nabla^l\maxpot\bar\maxw_{bl}t_a - 
\maxw_{al}\bar\maxw_b{}^l t^2\right) = \nabla^bz (\errsr{appone})_{ab}
\]
is essentially an algebraic error term. 

Next we consider the third term on the right hand side of
\eqref{eq:appone}. We can replace $\preernst$ by $\ernstform$ using
\eqref{eq:preernsternst}, and have 
\begin{align*}
\ernstform_{cb}\Re \preernst_a{}^c &= \frac{1}{\mu}\ernstform_{bc}
\Re(\bar\maxpotr\ernstform_a{}^c) + (\errsl{apptwo})_{ab} \\
\intertext{where}
(\errsr{apptwo})_{ab} &= \frac{2}{\mu} \ernstform_{cb}\Re(\bar\maxpot
\msmax_a{}^c)
\end{align*}
Now 
\[ \ernstform_{bc}\ernstform_a{}^c = \frac14 \ernstform^2 g_{ab} \]
and using that $\ernstform_{bc}\bar\ernstform_a{}^c$ is real valued,
we have
\begin{align*}
\ernstform_{bc}\bar\ernstform_a{}^c \nabla^bz &= \Im \ernstpot^{-2}
\ernstform_{bc}\bar\ernstform_a{}^c \ernstform^{db}t_d\\
&= -\frac14 \Im \ernstpot^{-2}\ernstform^2 \bar\ernstform_{ac}t^c\\
&= -\frac14 |\ernstpot|^4 \Im \left(
\ernstpot^{-4}\ernstform^2
\nabla_a\bar\ernstpot^{-1} \right)\\
& =  \frac14 |\ernstpot|^4 \Im(\ernstpot^{-4}\ernstform^2)\nabla_ay -
\frac14 |\ernstpot|^4 \Re(\ernstpot^{-4}\ernstform^2)\nabla_az\\
\intertext{here we can use \eqref{eq:identity1} and get}
& = \frac{|\ernstpot|^4}{t^2}\Im(\nabla\ernstpot^{-1})^2(\nabla_a y +
i\nabla_a z) -
\frac{|\ernstpot|^4}{4}\ernstpot^{-4}\ernstform^2\nabla_az
\end{align*}
so we get, using \eqref{eq:identityc1} from Corollary
\ref{cor:identities}, 
\begin{align*}
-\frac{1}{\ernstpot^2}\nabla^bz\ernstform_{cb}\Re\preernst_a{}^c &= 
- \frac{1}{\ernstpot^2}(\errsr{apptwo})_{ab}\nabla^bz -
  \frac{\bar\maxpotr}{8\mu \ernstpot^2}\ernstform^2 \nabla_a z \\
& \qquad 
+ \frac{\maxpotr \bar\ernstpot^2}{2\mu} \Im(\errsr{uno})
  \nabla_a\ernstpot^{-1} + \frac{\maxpotr \bar\ernstpot^2}{8 \mu
\ernstpot^4}\ernstform^2\nabla_az\\
&=
\frac{\ernstform^2}{4\mu\ernstpot^4}i\Im\left(\bar\ernstpot^2\maxpotr\right)\nabla_a
z  - \frac{1}{\ernstpot^2}(\errsr{apptwo})_{ab}\nabla^bz +
\frac{\maxpotr \bar\ernstpot^2}{2\mu} \Im(\errsr{uno})
\nabla_a\ernstpot^{-1}~.
\end{align*}

Next, we can consider adding in the second term on the right-hand side
of \eqref{eq:appone}, and expanding $\maxpotr =
\frac{\bar\kappa}{\mu}(\msmaxpot - \kappa\ernstpot) - \mu$ from the
definition, 
\begin{align*}
\frac{1}{4\ernstpot^3}\ernstform^2t^2\nabla_az &+
\frac{1}{4\mu\ernstpot^4}\ernstform^2
i\Im\left(\bar\ernstpot^2\maxpotr\right)\nabla_a z\\
&= (\errsr{uno} - \mu^{-2})\nabla_az\left[ \ernstpot
t^2 + \frac{i}{\mu} \Im\left(\frac{\bar\kappa}{\mu}\msmaxpot\bar\ernstpot^2 -
\frac{|\kappa\ernstpot|^2}{\mu}\bar\ernstpot - \mu \bar\ernstpot^2\right)\right]\\
\intertext{where $\errsr{uno}$ is as defined in Corallary
\ref{cor:identities}}
& = (\errsr{uno}-\mu^{-2})|\ernstpot|^2\nabla_az \left[
\bar\ernstpot^{-1} t^2 +
\frac{i}{\mu}\Im\left(\frac{\bar\kappa}{\mu}\frac{\bar\ernstpot}{\ernstpot}\msmaxpot
- \frac{|\kappa\ernstpot|^2}{\mu}\ernstpot^{-1} -
\mu\frac{\bar\ernstpot}{\ernstpot}\right)\right]~.
\end{align*}
Noting that $\msmaxpot$ is controllable by Lemma
\ref{lem:sumpotsisweakerror}, and using \eqref{eq:identity2p} to
replace $t^2$, we have
\begin{align*}
\bar\ernstpot^{-1} t^2 &+
\frac{i}{\mu}\Im\left(\frac{\bar\kappa}{\mu}\frac{\bar\ernstpot}{\ernstpot}\msmaxpot
- \frac{|\kappa\ernstpot|^2}{\mu}\ernstpot^{-1} -
\mu\frac{\bar\ernstpot}{\ernstpot}\right) \\
&= (y-iz)\left(\frac{1}{\mu^2}|\msmaxpot - \kappa\ernstpot|^2 +
\ernstpot + \bar\ernstpot + 1\right) \\
& \qquad +
i\Im\left(\frac{\bar\kappa\bar\ernstpot}{\mu^2\ernstpot}\msmaxpot\right)
+ \frac{iz}{\mu^2}|\kappa\ernstpot^2| - i \Im\left(
\frac{(y+iz)^2}{y^2 + z^2}\right)\\
&= y \left(\frac{1}{\mu^2}|\msmaxpot - \kappa\ernstpot|^2 +
\ernstpot + \bar\ernstpot + 1\right) +
i\Im\left(\frac{\bar\kappa\bar\ernstpot}{\mu^2\ernstpot}\msmaxpot\right)
\\
& \qquad - iz - \frac{iz}{\mu^2}\left(|\msmaxpot - \kappa\ernstpot|^2
- |\kappa\ernstpot|^2\right) - iz \frac{(-2y)}{y^2+z^2} -
i\frac{2yz}{y^2 + z^2}~.
\end{align*}
The first term is purely real: recalling that for our purpose we are
interested in the imaginary part of this expression, its contribution
will appear with a factor of $\errsr{uno}$. The second and fourth
terms are controlled by Lemma \ref{lem:sumpotsisweakerror}; the last
two terms cancel. So essentially we are only left with the third term,
$-iz$. In other words, up to some controllable errors, the imaginary
part of the sum of the
second and third terms on the right-hand side of \eqref{eq:appone}
contributes $\mu^{-2}z|\ernstpot|^2\nabla_az$, which corresponds
precisely to the second term on the right-hand side of
\eqref{eq:nablaA}. 

Lastly, we deal with the first term on the right-hand side of
\eqref{eq:appone}. We directly compute that 
\begin{align*}
-\ernstpot \nabla_a \ernstpot^{-1}\nabla_b\ernstpot^{-1}\nabla^bz & =
|\ernstpot|^2(y\nabla_ay + z \nabla_az - iz \nabla_a y + iy
\nabla_a z)(\nabla_by\nabla^bz + i (\nabla z)^2) \\
&= |\ernstpot|^2i(y\nabla_ay + z\nabla_a z)(\nabla z)^2 \\
& \qquad - |\ernstpot|^2i(z\nabla_a y - y\nabla_a
z)\frac{t^2}{2}\Im\errsr{uno} + \text{real-valued terms}~.
\end{align*}
The first of the terms corresponds to the first term on the right-hand
side of \eqref{eq:nablaA}, and the second term gives the error. 

So, collecting everything into one expression, we have that
\begin{align}
\notag \nabla_a\mathfrak{A} &= \frac{4\mu^2}{|\ernstpot|^2}\nabla^bz\Im\left(
\frac{t^c}{\ernstpot^2\bar\maxpotr}\nabla_a\msmax_{cb} -
\frac{\mu}{\ernstpot^2\bar\maxpotr}\msweyl_{dacb}t^ct^d\right)  + 2
\nabla_az
\Im\left(\frac{\bar\kappa\bar\ernstpot}{\mu^2\ernstpot}\msmaxpot\right)
\\
& \qquad + \mu^2 t^2(z\nabla_a y - y\nabla_a z)\Im\errsr{uno} - \Im\left[\frac{2\errsr{uno}\mu^2}{|\ernstpot|^2}\nabla_az\left(\ernstpot
t^2 + \frac{i}{\mu}\Im(\bar\ernstpot^2\maxpotr)\right)\right] \\
\notag & \qquad - \frac{z\nabla_az}{\mu^2}(|\msmaxpot - \kappa\ernstpot|^2 -
|\kappa\ernstpot|^2)  +
\Im\left[\frac{4\mu^3}{|\ernstpot|^2\ernstpot^2\bar\maxpotr}\nabla^bz
(\errsr{appone})_{ab}\right] \\
\notag & \qquad +
\Im\left[\frac{4\mu}{|\ernstpot|^2\ernstpot^2}\ernstform_{cb}\Re(\bar\maxpot\msmax_a{}^c)\nabla^bz
- \frac{\mu\maxpotr\bar\ernstpot}{\ernstpot}\Im(\errsr{uno})\nabla_a\ernstpot^{-1}\right]
\end{align}

\section{A mountain pass lemma}\label{appendix2}
The mountain pass theorem is perhaps most well known for its
application in calculus of variations in the form given by Ambrosetti
and Rabinowitz \cite{AmbRab1973}; but a finite dimensional version
goes back at least to Courant in 1950 \cite{Couran1977}. Here we give
(for not being able to find the exact statement needed elsewhere) a
version that is similar in statement to Katriel's topological mountain
pass theorem \cite{Katrie1994} but with a proof following Jabri
\cite[Chapter 5]{Jabri2003} and Nicolaescu \cite[Chapter
2]{Nicola2007a}. 

\begin{lem}\label{lem:mpt}
Let $\bar{S}$ denote a (possibly non-compact) finite dimensional
connected smooth paracompact manifold with boundary, with $S$ its 
interior and $\partial S$ 
the (possibly empty) boundary. Suppose we are given 
$f \in C^\infty(S,\Real) \cap C^0(\bar{S},\Real)$ such
that $f^{-1}((-\infty,a])$ is compact for any $a \in \Real$ (the
Palais-Smale condition). Suppose further that
there exists two real values $s_- < s_+$ and a closed subset
$C\subsetneq S$ such that
\begin{itemize}
\item $f|_{\partial S} \leq s_-$;
\item $f|_C \geq s_+$;
\item $C$ separates $\bar{S}$ with at least two of the connected
components intersecting $\{f \leq s_-\}$. 
\end{itemize}
Then $f$ attains a critical point in $S$ where the critical value is
at least $s_+$.  
\end{lem}
\begin{proof}
Let $S_1, S_2$ be two components of $\{f\leq s_-\}$ separated by $C$
(in the sense that every connected set containing both $S_1$ and $S_2$
must intersect $C$; the pair is guaranteed to exist by assumption).
Consider the collection $\Gamma$ of compact, connected subsets of
$\bar{S}$ that contains $S_1\cup S_2$. Let $m:\Gamma \to \Real$ be
defined by $m(T) = \sup_T f$. Let $(T_n)$ be a minimising sequence for
$m$ on $\Gamma$. Observe that since each $T_n\cap C \neq \emptyset$ necessarily $m(T_n)
\geq s_+$. Noting that $\overline{\cup_{j = k}^\infty T_j} \subset \{
f \leq m(T_k)\}$ is a closed subset of a compact set, the limiting set 
$T_\infty = \cap_{k = 1}^\infty
\overline{\cup_{j = k}^\infty T_j}$ is compact as the intersection of
a decreasing family of compact sets, and we have that 
\[ s_+ \leq m(T_\infty) \leq m(T) \quad \forall T\in\Gamma~.\]

Let $W = \{ x\in T_\infty: f(x) = m(T_\infty) \}$, we show that $W$
contains a critical point using gradient flow: fix,
once and for all, a smooth Riemannian metric $g$ on $S$. Then as $W$
is compact, $|df|_g$ attains a minimum $\alpha$ on $W$. If $\alpha =
0$ we are done. Suppose $\alpha \neq 0$, let $\eta$ be a non-negative 
 bump function supported inside $\{ 2m(T_\infty) > f > (s_+ + s_-)/2, |df|_g >
\alpha /2 \}$ with $\eta|_W = 1$. Then under the flow of $-\eta \nabla
f$, $T_\infty$ is mapped to another connected compact subset $T'$ of
$\bar{S}$. Since $-\eta\nabla f$ vanishes on $S_1, S_2$, the set
$T'\in \Gamma$. But since $ - \eta|\nabla f|_g^2 \leq 0$ and $-\eta
|\nabla f|_g^2|_W \leq -\alpha^2 < 0$, we have that the flow strictly
decreases $m$, that is $m(T') < f(W) = m(T_\infty)$, which leads to a
contradiction. 
\end{proof}

\bibliographystyle{amsalpha}

\providecommand{\bysame}{\leavevmode\hbox to3em{\hrulefill}\thinspace}
\providecommand{\MR}{\relax\ifhmode\unskip\space\fi MR }
\providecommand{\MRhref}[2]{%
  \href{http://www.ams.org/mathscinet-getitem?mr=#1}{#2}
}
\providecommand{\href}[2]{#2}

\end{document}